\documentclass[sigconf,  screen, acmthm, nonacm]{acmart}

\def\BibTeX{{\rm B\kern-.05em{\sc i\kern-.025em b}\kern-.08emT\kern-.1667em\lower.7ex\hbox{E}\kern-.125emX}}

\usepackage[shortlabels]{enumitem}
\usepackage[T1]{fontenc}
\usepackage{graphicx} % Required for inserting images
\usepackage{libertine}
\usepackage{multirow}
\usepackage{algorithm2e}
\usepackage{tikz}
\usepackage{pgfplots}

\title{Fast and Accurate Homomorphic Softmax Evaluation}
\date{January 28, 2023}
\author{Wonhee Cho}
\affiliation{\department{Department of Mathematics}\institution{Seoul National University}\city{Seoul} \country{South Korea}}
\email{wony404@snu.ac.kr}
\author{Guillaume Hanrot}\authornote{Corresponding author.}
\affiliation{\institution{CryptoLab, Inc.}\city{Lyon}\country{France}}
\email{guillaume.hanrot@cryptolab.co.kr}
\author{Taeseong Kim}
\affiliation{\department{Department of Mathematics} \institution{Seoul National University} \city{Seoul}\country{South Korea}}
\email{kts1023@snu.ac.kr}
\author{Minje Park}
\affiliation{\institution{CryptoLab, Inc.} \city{Seoul}\country{South Korea}}
\email{minje@cryptolab.co.kr}
\author{Damien Stehl\'e}
\affiliation{\institution{CryptoLab, Inc.} \city{Lyon} \country{France}}
\email{damien.stehle@cryptolab.co.kr}
\authorsaddresses{Corresponding author: Guillaume Hanrot.}
\def\R{{\mathbb{R}}}
\def\Z{{\mathbb{Z}}}

\def\SM{{\mathrm{Softmax}}}
\newtheorem{remark}{remark}
\newtheorem{assumption}{Assumption}
\let\svthefootnote\thefootnote
\newcommand\freefootnote[1]{%
  \let\thefootnote\relax%
  \footnotetext{#1}%
  \let\thefootnote\svthefootnote%
}
\newcommand{\style}[1]{\ensuremath{\mathsf{#1}}}

\newcommand{\pk}{\style{pk}}
\newcommand{\sk}{\style{sk}}
\newcommand{\evk}{\style{evk}}
\newcommand{\rk}{\style{rk}}
\newcommand{\add}{\style{add}}
\newcommand{\mult}{\style{mult}}
\newcommand{\BTS}{\style{BTS}}
\newcommand{\rot}{\style{rot}}
\newcommand{\enc}{\style{enc}}
\newcommand{\dec}{\style{dec}}
\newcommand{\preci}{\style{Prec}}

\newcommand{\C}{\mathbb{C}}
\newcommand{\m}{{\bf m}}
\newcommand{\w}{{\bf w}}
\newcommand{\ct}{\style{ct}}

\SetKwInput{KwData}{Input}
\SetKwInput{KwResult}{Output}
\copyrightyear{2024}
\acmYear{2024}
\setcopyright{acmlicensed}

\begin{document}

\begin{abstract}
Homomorphic encryption is one of the main solutions for building
secure and privacy-preserving solutions for Machine Learning as a
Service, a major challenge in a society where AI becomes more and more
pervasive. This motivates the development of homomorphic algorithms for the
main building blocks of AI, typically for the components of the various
types of neural networks architectures. 

Among those components, we focus on the Softmax function, defined by
$\SM(\mathbf{x}) = \left(\exp(x_i) / \sum_{j=1}^n \exp(x_j)
\right)_{1\le i\le n}$. This function is deemed to be one of the most
difficult to evaluate homomorphically, because of its multivariate
nature and of the very large range of values for $\exp(x_i)$.  The
available homomorphic algorithms remain restricted, especially in
large dimensions, while important applications such as Large Language
Models (LLM) require computing Softmax over large dimensional
vectors. Our algorithm has strong scalability properties in terms of
range and dimension while maintaining very good numerical accuracy. In
terms of multiplicative depth of the computation (a suitable measure
of cost for homomorphic algorithms), our algorithm achieves $O(\log
n)$ complexity for a fixed range of inputs, where $n$ is the Softmax dimension.

Our algorithm is especially adapted to the situation where we must
compute many Softmax at the same time, for instance, in the LLM
situation. In that case, assuming that all Softmax calls are packed
into $m$ ciphtertexts, the asymptotic amortized multiplicative depth
cost per ciphertext is, again over a fixed range, $O(1 + m/N)$
for $N$ the homomorphic ring degree (typically $N=2^{16}$, so that we have
$N \gg m$ in practice).

The main ingredient of our algorithms is a normalize-and-square
strategy, which manages to interlace the (numerically unstable)
exponential computation over a large range and (very expensive)
normalization, decomposing both in stabler and cheaper smaller
steps.

We have implemented our algorithms using the HEaaN implementation of
the CKKS HE system.  Comparing ourselves to the state of the art, our
experiments show, in practice, a gain of a factor~2.5 to 8 compared to
state of the art solutions.

These experiments demonstrate good accuracy (around 16-bit precision
in the worst case, around 20 on average) and support the linear
behavior in the dimension. The many-ciphertexts version allows us to
compute 8192 Softmax of dimension 256 in parallel in~486s
(single-thread CPU), corresponding to an amortized 0.06s per Softmax
call. All Softmax calls of the 32-layers LLaMa large language model (7B
version) with context length 128 on an RTX-6000 GPU take around 1.5
minutes, and the final Softmax call in dimension 32768 for token
generation takes less than 3 seconds. This suggests that
near-practicality may be accessible with dedicated hardware.

%with decent accuracy over $[-256, 0]$, in 351s.
\end{abstract}

%\begin{CCSXML}
%<ccs2012>
%<concept>
%<concept_id>10003752.10003809</concept_id>
%<concept_desc>Theory of computation~Design and analysis of algorithms</concept_desc>
%<concept_significance>500</concept_significance>
%</concept>
%<concept>
%<concept_id>10010147.10010257.10010293.10010294</concept_id>
%<concept_desc>Computing methodologies~Neural networks</concept_desc>
%<concept_significance>500</concept_significance>
%</concept>
%<concept>
%<concept_id>10002978.10002979.10002981.10011745</concept_id>
%<concept_desc>Security and privacy~Public key encryption</concept_desc>
%<concept_significance>500</concept_significance>
%</concept>
%</ccs2012>
%\end{CCSXML}

%\ccsdesc[500]{Theory of computation~Design and analysis of algorithms}
%\ccsdesc[500]{Computing methodologies~Neural networks}
%\ccsdesc[500]{Security and privacy~Public key encryption}


\maketitle

\freefootnote{\copyright 2024 Copyright is held by the authors. This is the author's version of the work. It is posted here for your personal use. Not for redistribution. The definitive version was published in CCS'24, \url{https://doi.org/10.1145/3658644.3670369}}

\section{Introduction}
The fast development of AI technologies in everyday life raises more and more pressing privacy and security concerns. In view of these concerns, Privacy-Preserving Machine Learning has emerged as a major research question over the last few years. We refer to \cite{xu2021privacypreserving} for an overview of the current problems, techniques, and solutions. 

In the context of the rapidly expanding "Machine Learning as a Service" (MLaaS) model, Homomorphic Encryption appears to be a promising solution. Initiated by the seminal work of Gentry~\cite{Gen09}, Homomorphic Encryption (HE) is a cryptographic technique that has since then developed at a fast pace. It allows a server to execute an AI algorithm, typically neural network-based, on a user's encrypted data and produce an encrypted answer while having no information neither on the actual data nor on the answer. Only the user can access the result by decrypting.   

On the bright side, homomorphic encryption is very versatile and, from
a theoretical point of view, applies "out of the box" to any AI
algorithm; on the dark side, despite major progress over the last
decade, one of the drawbacks of homomorphic encryption is its high
computational cost, making efficient algorithmic design a key to its
relevance for a given problem.  The deployment of homomorphic
encryption for a secure and private MLaaS model thus requires profound
algorithmic work on the ML primitives in order to design efficient and
adapted solutions. 

It is thus of importance to develop HE-tailored algorithms for the
most common Machine Learning building blocks, a task which has made a
lot of progress over the last few years regarding linear algebra~(see
e.g.~\cite{JKLN, GAZELLE}), convolutional steps~(see \cite{neujeans}
for a recent reference), and univariate activation functions (see
\cite{polapprox} for a recent reference), demonstrating the relevance
of HE as a practical solution for privacy-preserving MLaaS. 

\subsection{Context and Related Work} 
The implementation of neural networks in the homomorphic encryption
context raises several issues. Inference alternates linear steps
(linear algebra, convolutions) and non-linear steps (activation
functions, normalization steps). In this work, we focus on non-linear
steps and, more precisely, on the Softmax step, which probably stands
as the most difficult, being multivariate and involving very large or
very small intermediate results.

Given that AI algorithms rely on real numbers, the most convenient choice among HE schemes is CKKS\cite{CKKS}, which is fully homomorphic (FHE, i.e., allows for arbitrarily complex calculations), and natively enables homomorphic approximate addition, subtraction, and multiplication on real or complex numbers. Additionally, the large inner dimensions of AI primitives makes the fully SIMD mode of operation of CKKS particularly desirable in that setting. 

\subsubsection{Common strategies for activation functions in HE}
Several strategies are possible concerning activation function evaluation. As HE algorithms offer only approximate additions and multiplications (and rotations, to be defined later) as primitives, the starting point is to replace the non-polynomial activation function with a polynomial approximation of it. This requires \emph{a priori} information on the input range of the function, and, depending on the function under study, can be quite a formidable task. For instance, the ReLU function requires~\cite{Bern1913} a degree $O(2^p)$ polynomial to be evaluated to $p$ bits of precision over $[-1, 1]$. Smoother functions are somewhat less difficult to handle but still give rise to large degree polynomial approximations, which require deep circuits, and hence necessitate "bootstrapping" (see Section~\ref{sec:CKKS})-- an operation which remains, despite a lot of progress, the most expensive HE operation. 

Facing this difficulty, we can distinguish two strategies. 

The first one consists in designing and using HE-friendly functions, in order to replace the function itself with a simpler adaptation or even a small degree polynomial; sometimes the small degree polynomial is a decent polynomial approximation of the function under study, sometimes it is simply an unrelated small degree polynomial, like in~\cite{AESPA} where the $x^2$ polynomial is used as a replacement to the ReLU function. One drawback of this approach is that it requires retraining to achieve good accuracy for the neural network (if at all possible with the chosen function). 

The second one consists in developing HE-efficient algorithms for
existing functions, using various strategies in order to try to find
the right compromise between efficiency and accuracy; the
approximation might be generic, or finely tuned to the properties of
the neural network to be evaluated. We refer, for example,
to~\cite{LeeEtAl2023,polapprox}. 

It remains that, contrary to the first approach, this second setting
requires \emph{a priori} knowledge of the range; in between the two
approaches, the authors of \cite{ZimermanEtAl2023} show that training
can be modified to force a sharper range, allowing for a reasonably
accurate implementation by a small-degree polynomial. This, however,
again requires retraining.

\subsubsection{The Softmax function}
In this landscape, homomorphically evaluating the Softmax function,
defined, for $\mathbf{x}\in \R^n$ by
\[
\SM(\mathbf{x}) = \left(\frac{\exp(x_i)}{\sum_{j=1}^n \exp(x_j)}\right)_{1\le
  i\le n}
\] stands as one of the difficult problems to be solved,
because of its multidimensional nature, but also due to the fact that
its computation, even in plaintext, raises numerical issues which are
bound to be even harder in the HE computational model, as the values
of $\exp(x_i)$ can be either very large or very small. 

The HE implementation of Softmax has been subject to the two families
of approaches described above. The first direction was considered
in~\cite{Softmax2}, where a sigmoid function is used instead of the
Softmax function and as an approximation to it in the multi-class
classification problem. In~\cite{Softmax1}, all input values are
divided by a constant in order to handle numerical overflows, and a
Gumble softmax function is obtained. Finally, in~\cite{PrivFT}, a
quadratic polynomial is used as an approximation to Softmax.

In the second direction, the authors of~\cite{BMC} use an iterative
"square-and-normalize" approach akin to ours but using an \emph{a
priori} normalization, which severely limits its range and
accuracy. Finally, the state of the art of this second
approach~\cite{HETAL} follows the standard non-homomorphic strategy
for computing an accurate Softmax, namely first computing and
subtracting the maximum value of the input vector to all coordinates;
HETAL combines this idea with Domain Extension
Polynomials~\cite{CheonetAl2022} in order to achieve a large input
range for the exponential function. However, comparison is a
notoriously difficult task in HE and overall implies a deep circuit;
the same is also true of the final evaluation of the $\exp$ function
over a large range. 

\subsection{Technical contributions}

\subsubsection{Algorithmic strategy and asymptotic complexity}
Let $M$ be a positive real number. For $\mathbf{x}\in[-M, 0]^n$, we use the identity 
\[
\SM(\mathbf{x}/2^{j-1})_i = \frac{\SM(\mathbf{x}/2^{j})_i^2}{\sum_{t=1}^{n} \SM(\mathbf{x}/2^{j})_t^2}, 
\]
which follows from the definition, to build an iterative algorithm.
If~$k$ is an integer such that $k \approx \log M - \log \log n$, we
reduce a difficult (and numerically unstable) inversion over
$[1/n^{2^k}, 1]$ and a hard evaluation of $\exp$ over the large
interval $[-M, 0]$ to an evaluation of $\exp$ over $[-\log n, 0]$
followed by $k$ squares and easier and stable (inverses or) inverse
square-roots over $[1/n, 1]$, using the iteration described in
Algorithm~\ref{alg:main_loop}.

\RestyleAlgo{ruled}
\SetKwComment{Comment}{/* }{ */}
\SetKw{Return}{return}
\begin{algorithm}[ht]\caption{\label{alg:main_loop}Main loop of our Softmax algorithm}
$\lambda_j \gets \left(\sum_{i=1}^n {y^{(j-1)}_i}^2\right)^{-1/2}$ \; 
$z_i \gets \lambda_j\cdot y^{(j-1)}_i$, $i = 1 \ldots n$ \Comment*[r]{[Normalization]}
$y^{(j)}_i \gets z_i^2$, $i = 1 \ldots n$\Comment*[r]{[Squaring]}
\end{algorithm}

When $M < \log n,$ our algorithm becomes the naive strategy of
evaluating the exponential, then normalizing; we shall thus assume $M
\ge \log n$ in the sequel.

For the analysis of our algorithms, we shall use a model that looks simultaneously towards circuit complexity (\emph{multiplicative depth}) and ordinary complexity (\emph{total number of operations}). 
The \emph{multiplicative depth} is primarily important as it is directly related to the total number of bootstrappings, which are by far the most expensive HE operations. We obtain the following.
\begin{theorem}\label{th:intro_1}
Let $M \ge \log n$. On input ${\mathbf x}\in[-M, 0]^n$, Algorithm~\ref{alg:Softmax} returns an approximation of $~\SM({\mathbf x})$ using $(\log n) (\log M) (1 + o(1))$ multiplicative levels and $O((\log M) \sqrt{n\log n})$ multiplications. 
\end{theorem}
For fixed $M$, our level usage has a (quasi\footnote{recall that $M
\ge \log n$})linear dependency in $\log n$. We argue that,
heuristically, the actual level usage is $(\log n) (\log M)/2 (1 +
o(1))$. In order to compare ourselves to the state of the art, we
provide, in~Appendix~\ref{app:hetal}, a short analysis of HETAL. This analysis
leads to a heuristic multiplicative depth $(\log n) (\log M) (1 +
o(1))$.

Our precision analysis bounds the loss of accuracy (compared to the
internal precision of CKKS) by $\approx (\log n)(\log M)/2$ bits;
however, we provide heuristic arguments that this bound is very
pessimistic and that the loss of accuracy should be at most
$log M + 3(\log n)/2$ bits in the worst case. Our experiments show an even
better behaviour than this. 

We observe that our algorithm has two threads: a "wide" thread related
to the exponential computation that operates on all inputs in parallel
and a "thin" thread related to the inverse square root computation,
which operates on a single real number at a time.  This observation
allows us to amortize the simultaneous evaluation of Softmax over a
large number of entries using the SIMD capability of CKKS. The
following theorem, which focuses on the cost as a function of 
the number of ciphertexts $m$, summarizes our results in that case.  We
assume that the $m$ input ciphertexts encrypt $mN_0/n$ Softmax calls in
dimension $n$, where $m, n \le N_0$, the number of slots of our HE
system.  For our complexity measure, we divide each level used by the
number of ciphertexts to which it applies; we call the result
\emph{amortized level usage}.
\begin{theorem}
Let $m, n \le N_0$ two integers. On input $\ct_1, \dots, \ct_m$ containing $mN_0/n$ Softmax instances with inputs in $[-M, 0]$ in a packed format, the parallel version of our algorithm computes the corresponding Softmax with amortized level usage by ciphertext $2\log M + o_m(1)$, and amortized number of multiplications by ciphertext $O(\sqrt{\log n} + \log M) + o_m(1)$. 
\end{theorem}
We also describe a variant (see Algorithm~\ref{alg:Softmax1B} in Section~\ref{sec:parallel}) which reduces the amortized level usage by a factor of 2, from $\approx 2\log M$ to $\approx \log M$, up to an increase of the $\log M$ term to $(\log M)^2$ in the number of multiplications. 

Compared to this, a many-ciphertext version of HETAL has an amortized
level usage by ciphertext of the order of $\approx 4 \log M$, which is
4 times larger than our best solution. 

\subsubsection{Experimental results}
We conclude by an extensive experimental study demonstrating the
asymptotic behaviour expected in view of these theoretical analyses
for real-world values of $n, m, M$, with good accuracy ($\approx 16$
bits).

In particular, in \emph{single thread CPU}, the latency of our
algorithms in dimension~{128} to 1024 is around 250 seconds,
corresponding to a throughput of the order of 1 Softmax computation
per second. If 8192 Softmax of dimension 256 are to be computed
simultaneously, we obtain a latency of 414 seconds for a throughput of
$\approx 20$ Softmax per second. 

Finally, we demonstrate the solid scalability properties of our
algorithm by the computation of a Softmax of dimension 32768 over
[-256, 0]. The computation takes 254 sec, again on a single-thread
CPU, with decent accuracy.

Compared to the state-of-the-art solution HETAL~\cite{HETAL}, we
obtain a speedup ranging from 2.5 to 8, the best speedups being obtained
in the case where many Softmax need to be computed in parallel. 

\section{PRELIMINARIES}\label{sec:preliminaries}

\noindent{\bf Notation.}
We let vectors be denoted in lower-case bold face. For a real number $r$, we let $\lceil r \rceil$ denote the smallest integer that is no smaller than $r$.
The notation $\log$ stands for base-2 logarithm, whereas $\ln$ refers to the logarithm in base $e \approx 2.71...$. 
We let the Hadamard multiplication, namely entry-wise multiplication of vectors, be denoted by $\odot$: formally, given two $n$-dimensional vectors ${\bf v}= (v_1,v_2,\cdots,v_n)$ and $\w = (w_1,w_2,\cdots,w_n)$, ${\bf v} \odot \w$ denotes $(v_1\cdot w_1, v_2\cdot w_2, \cdots, v_n\cdot w_n)$. 

In our analyses, we shall use the $o(\cdot)$ and $O(\cdot)$ Landau
notations.

\subsection{A quick overview of Homomorphic Encryption}
Homomorphic Encryption (HE) is a cryptographic primitive which allows
to compute on encrypted data. More precisely, compared to the
classical setting of public-key encryption, HE provides one further
key, the \emph{evaluation key}, and one further operation, the
\emph{evaluation function} \textsf{Eval}, which with the help of the
evaluation key is able to evaluate functions from encrypted input
and obtain encrypted output; note that the evaluation key does not
allow decryption.

In abstract definitions, the evaluation function is assumed to be able
to evaluate any circuit.  In the descriptions and implementations of
concrete schemes, it is most often replaced by a set of specific
functions for addition, multiplication, etc.

The HE landscape as of today appears mostly focused around three
systems: CGGI/TFHE \cite{CGGI}, performing exact computation on
booleans or very small integers; BFV/BGV \cite{Brakerski12, FV12,
  BGV14}, performing exact computations on integers or elements of
small finite fields; and CKKS \cite{CKKS}, performing approximate
computations on real or complex numbers. We shall use the latter
one in the sequel. 

\subsection{The CKKS scheme}\label{sec:CKKS}
The CKKS \cite{CKKS} homomorphic encryption algorithm is a HE cryptosystem that allows one to encode and work on real or complex numbers as messages and to perform approximate arithmetic on the encrypted messages. Given a ring-degree $N$, CKKS plaintext messages are elements of $\C^{N/2}$, and the ring-homomorphism property of CKKS is related to the ring structure of $(\C^{N/2}, +, \odot)$. The coordinates of a given message are called \emph{slots}.

In the sequel, we shall denote by $N_0 = N/2$ the number of slots of
our HE system. The notation $n$ and the word ``dimension'' will be
refer to the dimension of the input vector of the Softmax function,
see~Eq.~(\ref{eq:smdef}).

\def\Dec{{\mathsf{Dec}}}
\def\Rot{{\mathsf{Rot}}}
\def\Add{{\mathsf{Add}}}
\def\Mult{{\mathsf{Mult}}}
\def\Enc{{\mathsf{Enc}}}

\subsubsection{Functionalities}
CKKS provides the following functionalities: 
\begin{itemize}
\item {\bf Key generation}: Given a security parameter $\lambda$ and a subset $S \subset \{1, \cdots, N_0-1 \}$, returns a public key $\pk$, a secret key~$\sk$, an evaluation key $\evk$, including a set of rotation keys $\{\rk_i\}_{i \in S}$.
\item {\bf Encryption}: Given a public key $\pk$ and a message $\m \in \C^{N_0}$, outputs a ciphertext $\ct = \Enc(\pk, \m)$ encrypting  $\m$. 
\item {\bf Decryption}: Given a secret key $\sk$ and a ciphertext $\ct$ encrypting $\m$, outputs $\Dec(\sk, \ct) \approx m$.
\item {\bf Addition}: Given two ciphertexts $\ct_1$ and $\ct_2$ of $\m_1$ and~$\m_2$, outputs a ciphertext $\ct_{\add} = \Add(\ct_1, \ct_2)$ such that we have $\Dec(\sk, \ct_{\add}) \approx \Dec(\sk, \ct_1) + \Dec(\sk, \ct_2)$.
\item {\bf Multiplication}: Given an evaluation key $\evk$ and two ciphertexts $\ct_1$ and $\ct_2$, outputs $\ct_{\mult} = \Mult(\evk, \ct_1, \ct_2)$ such that $\Dec(\sk, \ct_{\mult}) \approx \Dec(\sk, \ct_1) \odot \Dec(\sk, \ct_2)$. 
\item {\bf Rotation}: Given a rotation key $\rk_i$ and a ciphertext $\ct$ such that $\Dec(\sk, \ct) = (m_0,\cdots, m_{N_0-1})$, outputs a ciphertext $\Rot(\ct, i)$ with the property that $\Dec(\sk, \Rot(\ct, i)) = (m_{i+1},\cdots, m_{N_0-1},m_0,\cdots,m_i)$. The rotation index $i$ is defined modulo $N_0$, namely $\Rot(\cdot , -i)$ means $\Rot(\cdot, N_0-i)$.
\end{itemize}

\subsubsection{Levels}
CKKS comes with a notion of \emph{multiplicative level}. Heuristically
speaking, the "quality" of a ciphertext degrades when it goes through
multiplications. To formalize this, attached to each ciphertext is a
nonnegative integer, its \emph{level}, representing the multiplicative
budget left for this ciphertext: if we multiply $\ct_1$ with level~$l_1$
and $\ct_2$ with level $l_2$, we obtain as a result $\ct$ with
level $\min(l_1, l_2) - 1$. Once a ciphertext attains a certain bottom
level, it must be refreshed by a procedure named \emph{bootstrapping}
(\BTS) \cite{BTS} before undergoing further computations. It then
recovers a full multiplicative budget.

In practice, one can describe HE algorithms either in a \emph{levelled
HE} model, where one considers\footnote{Though feasible in theory, for
large computations, this leads to computing with huge HE parameters,
which is generally inefficient.} that the underlying HE parameters
allow for a computation of sufficiently large multiplicative
depth. However, for large computations a more realistic analysis is
obtained by considering the \emph{Fully HE} model, where bootstrapping
has to be used at suitable places to restore the level budget. The
placement of bootstrapping steps is a question by itself; it is
strategic as the bootstrapping steps are, most of the time, by far the
most costly steps of the whole computation.

\subsubsection{Numerical aspects}
The ability of CKKS to natively manipulate complex vectors as messages makes it especially adapted for AI applications. It also means that CKKS  only provides one with approximate operations. In practice, the approximate identities above can be rewritten as upper bound estimates on 
\begin{align*}    
\| \Dec(\sk, \ct_{\add}) - \Dec(\sk, \ct_1) - \Dec(\sk, \ct_2) \|_\infty,  \\
\| \Dec(\sk, \ct_{\mult}) - \Dec(\sk, \ct_1) \odot \Dec(\sk, \ct_2) \|_\infty, 
\textrm{ etc.},
\end{align*}
where $\|(v_i)_{1\le i\le n}\|_\infty  := \max_{1\le i\le n} |v_i|$ stands for the largest modulus of a coordinate of the vector $v$. For the corresponding estimates, we refer the reader, e.g., to~\cite{CKKS}.

\def\op{\mathrm{op}} In order to reason on CKKS arithmetic and
numerical issues, we need to model the behaviour of the underlying
arithmetic.  Intuitively, CKKS encodes a real number $r$ as $\lfloor
r\Delta\rceil, $ where $\Delta$ is a large integer, called the scaling
factor. It makes the behaviour of CKKS arithmetic very close to that
of a {\em fixed-point system}, implying that from a numerical point of
view we can model addition as error-free and other operations
introduce an absolute error $\varepsilon$.  We shall adopt this error
model in the present paper.

\subsubsection{Functions as polynomials}\label{ssec:polapprox}
The fact that the only arithmetical functions available in an HE
context are addition and multiplication implies that the set of
functions that one can evaluate is, in essence, restricted to
polynomials. A consequence of that fact is that the evaluation of more
general functions (we shall require exponential and inverse square
root) is performed through the use of \emph{polynomial
approximations}. Given a function $f$ over a fixed input interval $I$
and a degree, one finds, using various techniques (Chebyshev
interpolation, $L^2$ approximation, minimax approximation) a
polynomial $P$ such that for all $x\in I$, $|P(x) - f(x)| <
\varepsilon$.  This reduces the evalution of a function $f$ to the
evaluation of the polynomial $P$.

We shall often use the following facts:
\begin{itemize}
\item The evaluation of a polynomial $P$ of degree $d$ can be performed
  using $\lceil \log (d+1) \rceil$ multiplicative levels;
\item The evaluation of a polynomial $P$ of degree $d$ can be performed
  using $O(\sqrt{d})$ multiplications between ciphertexts (see~\cite{PS73}).
\end{itemize}

The cost of evaluating a function is thus directly related to the
degree of the polynomial approximation, which itself depends on the
range, the ``regularity'' of the function and the quality requirements
for the approximation, namely the value of $\varepsilon$.

\subsubsection{Computational model}
All these remarks allow us to define a computational model that we shall use to describe and analyze our algorithm: 
\begin{itemize}
    \item Our computation model implements addition, Hadamard multiplication, and rotations on vectors of $N_0$ slots; 
    \item Numerically speaking, all computations are performed in fixed point arithmetic with $p$ bits of precision following the binary point.
    \item Our model provides us, via polynomial approximation, with approximations of the functions $\exp$ and $x^{-1/{2^k}}$ for any fixed~$k$.
\end{itemize}
In Appendix~\ref{app:aux_functions}, we discuss how to build, in theory, those functions from our basic building blocks,
with complexity estimates. The practical implementation of those functions is somewhat different from this theoretical
study and is discussed in Section~\ref{sec:implementation}.

\subsection{Complexity model}\label{ssec:complexitymodel}
Since our algorithm targets scalability to situations where many Softmax have to be computed at once in large
dimension, we have chosen to provide simultaneously asymptotic estimates (which demonstrate scalability) and
implementation results, in order to validate our choice of complexity measures and the asymptotic estimates.
For this analysis, we need to define a suitable measure of cost for an HE algorithm. 

In short, the folklore wisdom regarding HE is the following: 
\begin{itemize}
  \item Bootstrapping is, by far, the most time-consuming of all operations, even if a lot of work
\cite{BTS, BTS2, BTS3, BTS4, BTS5}, both on the algorithmic and on the implementation sides, has been and is
still undertaken to improve its efficiency.
\item Apart from this, ciphertext-ciphertext multiplications and rotations are roughly on par and
  are the most expensive elementary operations.
\end{itemize}

Orders of magnitude of the cost of those operations for typical CKKS parameters can be found in Section~\ref{sec:implementation}. 

In view of this, we use level consumption, which is directly linked to
the number of bootstrapping calls, as a primary measure of
complexity. This is an imperfect model, but in our experience, it
gives a fair measure of the efficiency of HE algorithms requiring many
levels. We also provide counts for the number of multiplies and
rotations, which, combined with the previous, gives a good global
picture of the efficiency of our HE algorithms.

\section{The Baseline Softmax Algorithm}\label{sec:algorithms}
In this section, we present and analyze our algorithm for Softmax. 
For commodity we shall give a non-HE description of our algorithms, not addressing specific HE aspects. 

\subsection{Definitions, notations, elementary properties}
We start by reviewing the definition and useful properties of the $\SM$ function. 
\begin{definition}
Let $n$ be an integer; for $\mathbf{x} = (x_i)_{1\le i\le n} \in \R^n$, we define 
\begin{align}
\SM(\mathbf{x}) & = \lambda^{-1} \left(\exp(x_1), \dots , \exp(x_n)\right) \in \R^n, \label{eq:smdef}
\end{align}
 \textrm{where } $\lambda = \sum_{i=1}^n \exp(x_i)$.
\end{definition}

We mention two elementary properties of $\SM$ which are helpful in the discussions to come.
\begin{lemma}
For all  $\mathbf{x} = (x_i)_{1\le i\le n} \in \R^n$, the following identities hold:
\begin{enumerate}[a.]
\item We have, for all $\mu \in \R$,  
  \begin{equation}
  \SM((x_i - \mu)_{1\le i\le n}) = \SM(\mathbf{x}); \label{eq:normalized_sm}
  \end{equation}.
\item\label{ref:approx_normalization} If $\mathbf{y} = \alpha\cdot \SM(\mathbf{x})$, $\alpha \ne 0$, then \begin{equation}\label{eq:approx_renorm}
\SM(\mathbf{x})= \mathbf{y} / \sum_{i=1}^n y_i;
\end{equation}
\end{enumerate}
\end{lemma}

The intuition behind Item~b. is that if we have an algorithm that computes $\SM$ "up to a multiplicative constant," we can recover the actual $\SM$ thanks to a final normalization step. In other words, we may, to some extent, let errors accumulate as long as they accumulate multiplicatively and in the same way for all coordinates.
\subsection{Challenges}
We now discuss the numerical challenges associated with computing Softmax, especially in a fixed-point setting. These difficulties will shape our strategy. 
\subsubsection{Numerical issues}
A straightforward implementation of Equation~(\ref{eq:smdef}) raises
difficult issues, even if all inputs are known to lie in $[-M/2, M/2]^n$ for a fixed constant $M$~: 
\begin{enumerate}[(a)]
 \item Overflow, if $\exp(M/2)$ is too large to fit within the fixed-point specification; 
 \item Symmetrically, underflow if $\exp(-M/2)$ is too small to fit within the fixed-point specification; contrary to overflow, underflow is only an issue in the present context if the largest coordinates of the vector underflow; 
 \item In fixed point, any quantity which decreases to $\eta < 1$ loses $-\log \eta$ bits of relative precision, which cannot be recovered if the quantity grows back to its value. Securing precision in such a context requires taking steps to keep data "not too far away from 1" at all times.
 \item In the HE context, the normalization (division by $\lambda$) is notoriously difficult, unless one has good control over the interval where $\lambda$ lives; if we implement Eq.~\ref{eq:smdef} ``as is'', then we only know that $\lambda \in [n\exp(-M/2), n\exp(M/2)]$ which is typically a huge interval. 
\end{enumerate}

In the non-homomorphic setting, one avoids issues (a), (b) 
 by using Eq.~(\ref{eq:normalized_sm}) with $\mu = \max_{1\le i\le n} x_i$, which
ensures that the maximum component of the Softmax vector is equal to
$1$; it then does not overflow, and the other components can then
safely underflow. This also helps significantly with the more
HE-specific (d). Item (c) is taken care of automatically by
floating-point arithmetic.

HETAL~\cite{HETAL} chooses to use the same strategy for (a), (b);
however, comparisons are notoriously difficult in HE, and the estimation
of $M$ is the most-time consuming step, by far, of their Softmax algorithm.
It is then followed by the evaluation of the exponential function over
a large interval of width $M$, which requires domain extension functions
\cite{CheonetAl2022}. 

We shall use a (very) poor man's version of this idea, by subtracting a trivial upper bound for $\max_{1\le i\le n} x_i$: 
\begin{assumption}
  We assume that all our inputs lie within the interval $[-M, 0]$.
\end{assumption}
This assumption shall be made in the rest of this paper; we reduce to
this situation by using~Equation~(\ref{eq:normalized_sm}).  This
allows us to avoid issue~(a), but not the other ones. We now turn to
explain our strategy and how it allows to avoid (b), (c) and limit the
impact of~(d).

\subsubsection{Our strategy}
Rather than computing the full vector of exponential values
$(\exp(x_i))_{1\le i\le n}$ this way and then divide by the term $\sum_{i=1}^n
\exp(x_i)$, we interlace the two computations, using the classical
doubling identity $\exp(2x) = \exp(x)^2$ in the following way~: if
$\mathbf{y} = \SM(\mathbf{x})$ or $\mathbf{y} =
(\exp(x_i))_{1\le i\le n}$, then, for all $i$, 
 \begin{equation}\label{eq:doubling_sm}
 \SM(2\mathbf{x})_i = \frac{y_i^2}{\sum_{j=1}^{n} y_j^2}.
 \end{equation}

 Given a suitable parameter $k$, we shall compute 
 \[
 \mathbf{y}^{(0)} = \left(\exp(x_i/2^k)\right)_{1\le i\le n}
 \]
 directly, then use $k$ times this doubling identity in order to obtain $\mathbf{y}^{(k)} = \SM({\mathbf x})$. 

In this way, we trade one inversion for $k$ inversion steps;
however, we trade one \emph{hard} inversion for $k$ \emph{much easier}
ones. Indeed, the first number to invert is in $[n\exp(-M/2^k), n]$. For 
$k \approx \log M$, this is a rather small interval, which
makes the inversion tractable. The following lemma considers the 
  following loop iterations.
\begin{lemma}\label{le:cs}
  Let $\mathbf{x}\in \R^n$ and $\mathbf{y}
  =\SM(\mathbf{x}/2^j)$. Then, we have $\sum_{i=1}^n y_i^2 \in [1/n, 1]$.
  \end{lemma}
\begin{proof}
  This follows from $y_i \ge 0$,
  $\sum_{i=1}^n y_i = 1$ and Cauchy-Schwarz inequality.
\end{proof}

The factor $\lambda^2 := 1/(\sum_{i=1}^n y_i^2)$ is a
\emph{normalization factor}. Equation~(\ref{eq:doubling_sm}) then
leads to a "\emph{square-and-normalize}" strategy; for better numerical accuracy, 
we shall rather adopt a
"\emph{normalize-and-square}" strategy, which first computes $z_i =
\lambda y_i$, then squares.  This is
slightly less efficient in computational terms but guarantees that
\[
\max z_i\in [n^{-1/2}, 1],
\]
at all times. Only then do we compute $z_i^2$.
This allows us to limit the impact of difficulty~(c) --
a direct implementation of Equation~(\ref{eq:doubling_sm}) would
lead to numbers as small as $1/n^2$ after squaring and before
normalization, leading, in fixed point arithmetic, to a potential
loss of $2\log n$ bits of accuracy instead of $\log n$.

\begin{remark}\label{rem:approx_renorm_practice}
Thanks to Equation~(\ref{eq:approx_renorm}), our strategy has a
(partial) self-correcting ability: the multiplicative numerical
error in the computation of the normalization factors will be
corrected during the next normalization step. As the goal of normalization
is mostly to ``keep quantities close to 1'', this allows for a much
cheaper approximate normalization. We shall come back to this in Section~\ref{sec:implementation}.
\end{remark}

\begin{remark}
Depending on the internal precision and the target precision, the
square-and-normalize strategy can be preferred, as in practice, it may
lead to a slightly easier computation of the normalization factor
(inverse rather than inverse square root).  Other normalizations using
$k$-th powers with $k\ge 3$ can be considered; they are less efficient
asymptotically but can be of practical interest for small values of
$n$. We discuss this issue in Appendix~\ref{se:diffnorm}.
\end{remark}

\subsection{High-level description of the algorithm and precision analysis}
We now turn to a formal description of the algorithm. We shall first
describe our algorithm in an "ordinary", fixed-point setting, and
focus on the precision analysis. We will then explain how to implement
this ``ordinary'' algorithm homomorphically. 

\RestyleAlgo{ruled}
\SetKwComment{Comment}{/* }{ */}
\SetKw{Return}{return}
\begin{algorithm}[ht]
\caption{$\SM$}\label{alg:Softmax}
\KwData{$\mathbf{x} \in [-M, 0]^n, k \in \Z_{> 0}$}
\KwResult{$\mathbf{y} = \SM(\mathbf{x}).$}
$y^{(0)}_i \gets \exp(x_i/2^k), i = 1 .. n$ \; 
\For{$j$ \textrm{from} $1$ \textrm{to} $k$}{ 
$\lambda_j \gets \left(\sum_{i=1}^n {y^{(j-1)}_i}^2\right)^{-1/2}$ \; 
$z_i \gets \lambda_j\cdot y^{(j-1)}_i$, $i = 1 .. n$ \Comment*[r]{[Normalization]}
$y^{(j)}_i \gets z_i^2$, $i = 1..n$\Comment*[r]{[Squaring]}
}
\Return($y^{(k)}$)
\end{algorithm}

\begin{figure*}
  \includegraphics[width=\textwidth]{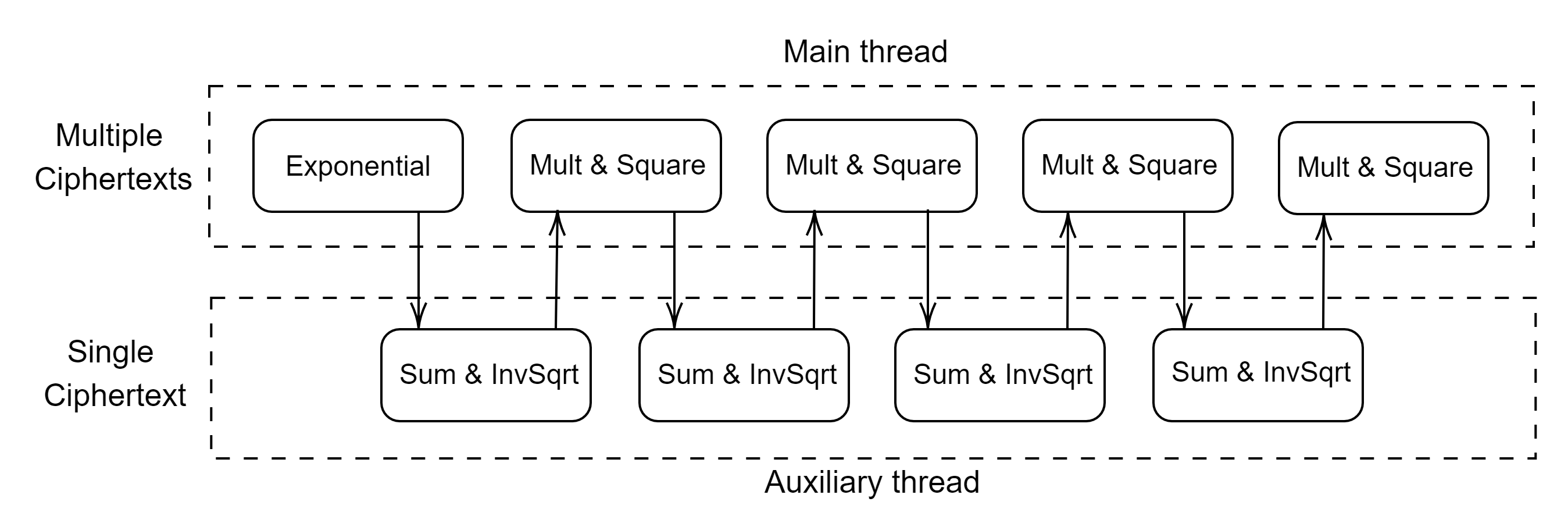}
  \caption{High-level Overview of the Softmax Algorithm. \label{fig:threads}}
  \Description{This image illustrates the two threads of our Softmax algorithm. The top thread computes exponential \& normalizes and squares, whereas the bottom thread computes normalization factors.}
  \label{fig:Meta-BTS}
\end{figure*}

We now give formal statements of (approximate) correctness. Recall
that we have fixed the internal precision of our computational model
to be $p$-bit fixed point.
\begin{assumption}
Addition incurs no error in our computational model. We assume that we
have native access to an exponential function
and an inverse square root function that return a value within $2^{-p}$
of the actual mathematical value. 
\end{assumption}
The first part of this assumption is standard in fixed-point models.
The second part is discussed in Appendix~\ref{app:aux_functions}. We
shall write $\varepsilon = 2^{-p}$.

Note that with exact computations, the correctness of both algorithms
follows from the doubling formulas given above. Controlling error
amplification leads to a somewhat technical discussion. The reader
mostly interested in the practical behaviour of the algorithms should
feel free to 
jump to the
implementation section which demonstrates the practical accuracy of
Algorithm~\ref{alg:Softmax} in typical cases.

\begin{theorem}\label{th:prec_softmax}
  Let $M$ be a fixed positive real number.
  We define $k =  \lceil \log M - \log \ln n\rceil$ and
  $\varepsilon = 2^{-p}$ and assume that
  \[\max(n^2, 2.9 (2.08 \sqrt{n})^k) \varepsilon \le 1/1000.\] 
Algorithm~\ref{alg:Softmax}, on input $\mathbf{x}\in[-M, 0]^n$, returns a vector $y\in [o(1),1+o(1)]^n$ such that
  \[
     \max_{1\le i\le n} |y_i - \SM(x)_i| \le 2.9 \cdot \left((2.08 \sqrt{n})^k (n+1) + 15.5n^2\right) \varepsilon.
  \]
\end{theorem}

\begin{proof} See Appendix~\ref{se:pfthm}.
\end{proof}

For any $\alpha > 2, \beta > 1$, for $n$ large enough, our proof can be adapted to yield
a bound where 
$2.09(2.08)^k$ is replaced by by $\beta \cdot \alpha^k$. Still, 
the error estimate is actually extremely pessimistic and far from being
observed in practice.

Looking at the proof in a more heuristic way, we are also able to
argue in Appendix~\ref{se:pfthm} that we should rather expect an error increase of~$O(n)$
for the normalize-and-square strategy, which, together with the
final normalization error (on average) yields a~$\approx \log M + 3(\log n)/2$
loss of precision overall.

\subsection{Description and cost analysis}
In this subsection, we give a description of the homomorphic
algorithm; this involves mostly explaining how to compute
the values $(\exp(x_i/2^k))_{1\le i \le n}$ at the first step, and then
$\lambda_j$ at each loop iteration $1\le j \le k$, as the other steps
can be performed directly using homomorphic addition and
multiplication. We refer to Subsection~\ref{ssec:polapprox} for a
general discussion of function evaluation through polynomial
approximation. In this section, we disregard the numerical error
issues, in particular in the context of polynomial evaluation.

We recall that we adopt an asymptotic viewpoint and analyze the
cost of our approach as a function of $n$ (the dimension) and $M$ (the
range). Our cost measures are the ones defined in
Section~\ref{ssec:complexitymodel}. 

\begin{assumption}
For this analysis, we assume that the number $N_0$ of slots of our HE
algorithm is greater than the Softmax dimension $n$. 
\end{assumption}

We shall let the dimension $n$ and the range $M$ tend to infinity
independently. We shall allow the precision $p$ to tend to infinity,
but no faster than $O(\log n)$ -- in view of the typical parameter
choices of CKKS this is a realistic assumption for large $n$.

\subsubsection{Auxiliary thread and main thread}
In order to prepare the description and the analysis of the parallel
version, we shall distinguish the computation of the normalization
factors $\lambda_j$, which we shall call the \emph{auxiliary thread};
we call the rest of the computation the \emph{main thread} (see
Figure~\ref{fig:threads} for an illustration).

It is natural to analyze separately those two parts, mostly because
the auxiliary thread is handling essentially one single real number at a
time, whereas the main thread operates over $n$ numbers at a
time. This distinction prepares the discussion of the case of "many
ciphertexts at once" from Section~\ref{sec:parallel}.

The full auxiliary thread, except the computation of the inverse square root, is described in~Algorithm~\ref{alg:AuxThread} below. We now analyze it. 
\begin{lemma}\label{prop:convert}
  Let $N_0$ be the number of slots of our HE implementation, and $n$ an
  integer dividing $N_0$. Starting from a ciphertext containing $x_1,
  \dots, x_n$ in slots $1, 1+N_0/n, \dots, 1 + (n-1)N_0/n$ and $0$ in the
  remaining slots, we can compute a ciphertext containing $\sum x_i^2$
  in slots $1, 1+N_0/n, \dots, 1 + (n-1)N_0/n$ using one multiplication
  and $\lceil\log n\rceil$ rotations; if $\lambda =
  1/\sqrt{\sum_{i=1}^n x_i^2}$, we can then compute a ciphertext
  containing $x_i/\lambda$ in slot $1 + iN_0/n$ using one call to the
  inverse square root function, two Hadamard products, and $2\lceil
  \log n\rceil$ rotations.
\end{lemma}
\begin{proof}
We use the procedure described in Algorithm~\ref{alg:AuxThread}.
\end{proof}
\def\ctI{\mathsf{ct_I}}
\def\ctO{\mathsf{ct_O}}
\def\mask{\mathsf{mask}}
\begin{algorithm}
\LinesNumbered
\caption{Auxiliary thread}\label{alg:AuxThread}
\SetKwComment{Comment}{/* }{ */}
\KwData{$\ctI$ = $(x_1, 0, \dots, x_2, 0, \dots, x_n, 0, \dots , 0)$ containing $x_i$ in slot $iN/n$}
\KwResult{$\ctO$ = $(1/\lambda, 0, \dots, 1/\lambda, 0, \dots, 1/\lambda, 0, \dots, 0)$, (nonzero values in slots $iN/n$), where $\lambda=(\sum_{i=1}^n x_i^2)^{-1/2}$.}
$\ctO \gets \ctI \odot \ctI$\; 
$\mask \gets (1, 0, ..., 0)$\; 
\For{$j$ \textrm{from} $0$ \textrm{to} $\lceil \log (n) \rceil -1$}{
    $\ctO \gets \ctO + \mathrm{Rot}(\ctO, -N/n \cdot 2^j)$ 
    }
$\ctO \leftarrow \mathrm{InvSqrt(\ctO)}$\;
$\ctO \leftarrow \ctO \odot \mask$\;
\For{$j$ \textrm{from} $0$ \textrm{to} $\lceil \log (n) \rceil -1$}{
    $\ctO \gets \ctO + \mathrm{Rot}(\ctO, N/n \cdot 2^j)$ 
    }   
\Return{$\ctO$}
   \end{algorithm}

Steps 7-10 might seem superfluous, as they deal with a blind spot of
our numerical model. In practice, when implementing the algorithm,
they are required to achieve a good accuracy in the worst case.
 
Our numerical model overlooks the fact that in CKKS, the error induced
by~Steps~3-5 is not the same in all slots, while the self-correcting
property of Algorithm~\ref{alg:Softmax} is limited to errors which are
the same in all slots. After Steps~3-5, we might observe significant
variations between slots, which are amplified by the inverse square
root computation. Without correction, this makes the overall precision
quite poor. Steps 7-10 maintain the same level of error, but average
most of the error over the slots -- the part of the error which is
averaged is then removed by the self-correcting capability of
Algorithm~\ref{alg:Softmax}.

\subsubsection{Asymptotic analysis of the main thread}
In levelled HE, the cost of the main thread for $\SM$ is: 
\begin{itemize}
\item one evaluation of $\exp$ over $[-M/2^k, 0]$; 
\item two multiplications at each loop iteration. 
\end{itemize}

\begin{theorem}\label{th:main_thread}
  On input $(x_i)\in [-M, 0]^n$, with parameter $k = \log M - \log \ln n +
  O(1)$, the main thread of Algorithm~\ref{alg:Softmax} 
  uses 
$  2\log M - \log\ln n + O(1)$
 levels and $  O\left(\sqrt{\log n} + \log M\right)$
  ciphertext-ciphertext multiplications.

\end{theorem}
\begin{proof}
  The loop costs 2 levels per iteration; as $k = \log M - \log \ln n + O(1)$, the estimate follows from the cost of evaluating an exponential over $[-M/2^k, 0] = [-O(\log n), 0]$; see Lemma~\ref{le:exp} in Appendix~\ref{app:aux_functions}. 
\end{proof}

Note that the cost implictly depends on the precision $p$ since we
assumed that $p = O(\log n)$.

\begin{table}
\[
\begin{array}{|c|c|c|}
\hline
  \textrm{Input range} & \exp & 
  \textrm{total depth (main thrd.)}
  \\
  \hline
  [-2^k, 0], k\ge 2  & 4 
  & 2k + 4   
  \\
  \hline 
\end{array}
\]
\caption{\label{tab:depth_main}Level usage of the main thread of Algorithm~\ref{alg:Softmax} in our experiments}
\end{table}

Table~\ref{tab:depth_main} gives the non-asymptotic, practical depth
of the computations for our experiments; in the latter, the depth of
the main thread did not depend on the Softmax dimension $n$, only on
the input range.  Our parameters were chosen in order to achieve
$\approx 16$ bits of accuracy.

\subsubsection{Asymptotic analysis of the auxiliary thread}\label{sse:aux_anal}
We now turn to the study of the auxiliary thread. We make heavy use of
Remark~\ref{rem:approx_renorm_practice}: in the
description of the algorithm, we allow ourselves to replace the
inverse square root by a rough approximation to it (i.e., we replace
$1/\sqrt{u}$ by $P(u)$ such that $ P(u)^2\cdot u \in [1/\theta,
  \theta]$, for $\theta$ a constant not too far from $1$ -- we will
write $2^{O(1)}$ instead of $\theta$ in the sequel).

In order to estimate the cost of this rough approximation to the
inverse square root, we need to make more explicit the input/output
intervals of each call to the approximate inverse square root. 
\begin{itemize}
\item Step 1: the input $\sum_{i=1}^n y_i^{(0)^2}$ is in
  $[n \exp(-M/2^{k-1}), n]$, and we compute an output such that $\lambda_0^2 (\sum_{i=1}^n y_i^{(0)^2})$ is not too far from $1$, or in asymptotic terms $\in [2^{-O(1)}, 2^{O(1)}]$. 
\item Step $j\in [2, k-1]$: the input is such that $\sum_{i=1}^n y_i^{(j)^2} \in
  [2^{-O(1)}/n, 2^{O(1)}]$ (see Lemma~\ref{le:cs}). Similarly, we expect to have $\lambda_j^2 \left(\sum_{i=1}^n y_i^{(j)^2}\right) \in [2^{-O(1)}, 2^{O(1)}]$.
\item Step $j = k$ is similar, except that we require 
  \[
  \lambda_k^2 \left(\sum_{i=1}^n y_i^{(k)^2}\right) \in [1-2^{-p}, 1+2^{-p}].
  \]
\end{itemize}
  
\begin{theorem}\label{th:aux_thread}
  For each iteration of the loop, the auxiliary thread of
  Algorithm~\ref{alg:Softmax} 
  requires $\log n (1 + o(1))$ levels and
  $O(\sqrt{n\log n})$ multiplications.
\end{theorem}
\begin{proof}
  The result follows from the discussion above and from Lemma~\ref{le:invsqrt} on the polynomial approximation of the  inverse square root function.
\end{proof}

As pointed in Appendix~\ref{app:aux_functions}, experimentally, using optimal polynomial approximation, the level consumption can be reduced to $(\log n)/2 (1 + o(1))$. This is illustrated by the level usage in Table~\ref{tab:depth_aux_thread}. Similarly, under the same heuristic assumption, the number of ciphertext-ciphertext multiplications should be of the order of $O(n^{1/4})$ rather than $O(\sqrt{n \log n})$.

\subsubsection{Practical values of the degrees of the polynomial approximations}
Table~\ref{tab:depth_aux_thread} states the depth of the computation for
various dimensions, using the polynomial degrees used in our
implementation rather than asymptotic estimations.  The total level
consumption is the total consumption of the~$k$ calls
to~Algorithm~\ref{alg:AuxThread}.

In practice, the first and the last calls differ from the other ones: for our parameter choices, the first one addresses a somewhat larger interval, while the last one needs to be precise, contrary to the previous ones (see Remark~\ref{rem:approx_renorm_practice}). 

\begin{table}
\[
\begin{array}{|c|c|c|c|c|}
 \hline
  {\textrm{range}} & \textrm{dim.} & \textrm{sq.\& mask} & x^{-1/2} & \textrm{depth} \\
  & $n$ & (\textrm{Steps}~1, 7) & (\textrm{Step}~6) & \textrm{(aux. thread)} \\
  \hline
  [-2^k, 0] & 128  & 2^{\times k} &6, 4^{\times (k-2)}, 7 & %
  6k+5 % & 5k+3 \\
  \\
  \hline
  [-2^k, 0] & 256  & 2^{\times k} & 6, 5^{\times(k-2)}, 7 & %
  7k+3 %
  \\
  \hline
  [-2^k, 0] & 512  & 2^{\times k} & 6, 5^{\times(k-2)}, 7 & 
  7k+3
  \\
  \hline
  [-2^k, 0] &  1024 & 2^{\times k} & 6, 6^{\times{(k-2)}}, 8 & 
  8k+2
  \\
  \hline
\end{array}
\]
\caption{\label{tab:depth_aux_thread}Level usage of polynomial approximations and total level usage for the auxiliary thread in our experiments. The meaning of the three figures in the column $x^{-1/2}$  is as follows: the first figure is the depth for loop iteration $1$, the notation $d^{\times(k-2)}$ means depth $d$ for the $(k-2)$ intermediate loop iterations, and the last figure is for loop iteration $k$. Similarly, the notation $d^{\times k}$ in column sq.\& mask means depth $d$ for all the $k$ loop iterations. }
\end{table}

\subsection{Overall evaluation and comparison with HETAL}
Using Theorems~\ref{th:main_thread} and~\ref{th:aux_thread}, we see that our total level consumption is, in all cases, 
\begin{align*}
    & \le (2 \log M + (\log M - \log \log n) \log n) (1 + o(1)) \\
    & \sim (\log M) (\log n), 
\end{align*}
when $n, M$ tend to infinity.

Again, the discussion concluding Section~\ref{sse:aux_anal} suggests that by using optimal polynomial approximations for the inverse square root, this estimate should be divided by 2, i.e. $(\log M) \frac{\log n}{2}$ (see the discussion following Theorem~\ref{th:aux_thread}. 

We now compare the algorithm to HETAL~\cite{HETAL}, the current state of
the art for HE-Softmax. We sketch an analysis of the latter in
Appendix~\ref{app:hetal}; this analysis suggests that HETAL level
usage is $\sim (\log n) (\log M)$.

This shows that our algorithm gains a factor of 2 regarding the level
consumption. It should however be noted that it comes at the cost of a
significant increase in the number of ciphertext-ciphertext
multiplications compared to HETAL ($O(n^{1/4}\log M)$ against
essentially $O((\log n) (\log M))$). We recall however that level
consumption is, asymptotically, a finer measure of cost as it is
directly related to the total bootstrapping cost.

\section{Computing many Softmax at once}\label{sec:parallel}
In this section, we discuss the situation where we have a large number $L$ of Softmax on inputs $({\mathbf x}^{\{\ell\}}_i)_{1\le i\le n}$, $1\le \ell \le L$, to be computed at the same time. 

\subsection{The one ciphertext case}

If we want to compute in parallel $L$ Softmax of dimension $n$ and
$Ln\leq N_0$, the intrinsic SIMD parallelism of CKKS, which can perform
simultaneously up to $N_0$ identical computation on different values,
is sufficient.

Indeed, in this case, we may notice that Algorithm~\ref{alg:AuxThread},
on input
\[(x_1^{\{1\}}, x_1^{\{2\}}, \dots, x_1^{\{L\}},\mathbf{0}_{N_0/n-L} ,x_2^{\{1\}},x_2^{\{2\}},\dots, x_n^{\{L\}},\mathbf{0}_{N_0/n-L}),
\] where we let $\mathbf{0}_j$ denote a sequence of $j$ zeros,
already returns the desired ciphertext containing
$x_i^{\{\ell\}} / \sqrt{\sum_{i=1}^n x_i^{\{\ell\}^2}}$ in slot $i + \ell n$. 

With this packing, there is thus nothing to modify, and we can take
full advantage of the SIMD capability of an HE system like CKKS: the
computation time is the same for $N_0/n$ Softmax as it would be for
a single one.

\subsection{The many ciphertexts case}
We now assume that $Ln \ge N_0$, and, for the sake of
simplicity, $L$ divides $N_0$ and that $N_0$ divides $Ln$. We can reduce to the
previous case by splitting into $Ln/N_0$ times $N_0/n$ Softmax of size $n$, but this is sub-optimal, as we only use $N_0/n$ slots of the ciphertext of the auxiliary thread. We let $m = Ln/N_0$ denote the number of ciphertexts of the main thread.

For $L \le N_0$, one should factor the auxiliary thread computation
coming from the various ciphertexts; this means that we pack all the
auxiliary threads into one single ciphertext as long as we compute at
most $N_0$ Softmax at once.

Algorithm~\ref{alg:AuxThread} can be adapted to that setting. Recall our choice
of packing: 
\begin{align*}
\textrm{ctxt}_j = & (x_{1 + jN_0/L}^{\{1\}}, x_{1 + jN_0/L}^{\{2\}}, \dots, x_{1+jN_0/L}^{\{L\}}, \\
& x_{2+jN_0/L}^{\{1\}}, \dots, x_{2+jN_0/L}^{\{L\}},\\
& \dots, \\
& x_{(j+1) N_0/L}^{\{1\}}, \dots x_{(j+1) N_0/L}^{\{L\}}), 0\le j\le m-1.
\end{align*}
It allows to have a total number of rotations that is equal to $(m+1) \log (N_0/L) = (m+1) \log (n/m)$ in that case.

A similar idea can be applied to HETAL: as each comparison layer
replaces two coordinates in a Softmax by their maximum, the slot usage
at each such layer is divided by a factor of 2. The remaining slots
can thus be packed into half the number of ciphertexts, until we obtain
one single ciphertext.

We introduce the \emph{amortized level consumption per ciphertext} as
a suitable measure of the computational cost, as it is directly
related to the total number of bootstrapings to be performed, which again
dominates the total cost. In order to simplify the statements, we
focus on the asymptotic dependency in $m$, the notation $o_m(1)$
meaning that we neglect lower order terms in $m$. 

\begin{theorem}\label{th:parallel_levels}
Assume that $M > \log n$, $N_0 \ge m, n$. On input $\ct_1, \ct_2,$ $\dots, \ct_m$ packing $L = mN_0/n$ Softmax of dimension $n$, the total amortized level usage per ciphertext of the many-ciphertext version of Algorithm~\ref{alg:Softmax} is 
\[
2\log M + \log \log n + O_m(1).%
\]

The amortized number of ciphertext-ciphertext multiplications per ciphertext is $O(\log M + \sqrt{\log n}) + O_m(1)$. 
\end{theorem}
Our amortized level cost is thus almost constant in $\log n$; 
in comparison, we estimate the amortized level usage of HETAL to be
$\ge 4 \log (M) + O_m(1)$, see Appendix~\ref{app:hetal}.

\subsection{A variant of the main algorithm}
We now propose a variant of~Algorithm~\ref{alg:Softmax}, which reduces
the level consumption of the main thread per loop iteration from 2 to
1. The total depth is then $k+1$ plus the depth for exponential
evaluation. The main weakness of this approach is the fact that its
implementation requires designing one auxiliary function ( $x \mapsto
x^{-1/2^j}$) for each loop iteration, making it more difficult to
implement. It also performs more multiplications overall, which makes
it less efficient in the case of a single ciphertext.

\begin{algorithm}
\caption{$\SM$ (version B)}\label{alg:Softmax1B}
\KwData{$\mathbf{x} \in [-M, 0]^n, k \in \Z_{> 0}$}
\KwResult{$\mathbf{y} = \SM(\mathbf{x}).$}
  $\lambda \gets 1$ \;
  $y^{(0)}_i \gets \exp(x_i/2^{k}), i = 1..n$ \; 
  \For{$j$ \textrm{from} $1$ \textrm{to} $k$}{ 
    $\lambda_j \gets \left(\sum_{i=1}^n {y^{(j-1)^2}_i}\right)^{-2^{-(j-1)}}$ \; 
    $\lambda \gets \lambda \cdot \lambda_j$ \;
    $z_i \gets \lambda\cdot y^{(0)}_i$, $i = 1..n$ \; 
    $y^{(j)}_i \gets z_i^{2^j}$, $i = 1..n$ \; 
  }
\Return $(y^{(k)})$\; 
\end{algorithm}

This modification has substantial impact in the many ciphertexts
case, as the amortized level consumption of the latter is
asymptotically equivalent to the level consumption of the main thread.

From a non-asymptotic point of view, if our range is such that the
total level usage fits into the level budget, this approach allows us
to totally avoid bootstrapping the main thread, which has a major
impact when the number of Softmax to be computed at once (and, thus,
the number of ciphertexts) is large. We illustrate this example in
Section~\ref{sec:implementation} for $M = 128$.

From an asymptotic point of view, we obtain the following result.
\begin{theorem}
With the notations and assumptions of
Theorem~\ref{th:parallel_levels}, on input $\ct_1, \dots, \ct_m$, with $N_0
\ge m, n$, the amortized level consumption by ciphertext of the
many-ciphertexts version of Algorithm~\ref{alg:Softmax1B} is $\log M
+ O_m(1)$; the amortized number of ciphertext-ciphertext 
multiplications is $O(\log^2 M + \sqrt{\log n}) + O_m(1)$.
\end{theorem}

  This result follows from the fact that the amortized level cost per
  ciphertext of the main thread is $\log (M/2^k) = \log \log n$ for
  computing the exponential, and $1$ for each of the $k = \log M - \log
  \log n + O_m(1)$ iterations. The amortized cost per ciphertext of
  the auxiliary thread is $o_m(1)$ under the assumption that $N_0 \ge
  m$. For the number of multiplications, the $\sqrt{\log n}$ term
  comes from the evaluation of $\exp$, while the $O((\log M)^2)$ term
  comes from the step $y_i^{(j)} \leftarrow z_i^{2^j}$ evaluated for
  $j$ from  $1$ to $O(\log M)$.

\section{Implementation results}\label{sec:implementation}
We have implemented Algorithms~\ref{alg:Softmax} and~\ref{alg:Softmax1B}. In this section, we describe our implementation and give experimental results supporting our claims regarding efficiency and accuracy. 

\subsection{Implementation of Algorithms~\ref{alg:Softmax} and~\ref{alg:Softmax1B}}
\subsubsection{Software and HE parameters}
Our implementation uses the C++ HEaaN library\footnote{CryptoLab HEaaN library, which is available at https://www.heaan.it}. We use the FGb parameter set, one of the preset HEaaN parameter sets. The corresponding parameters are described in Table~\ref{tab:base}, together with an estimate for the cost of the homomorphic operations. 

This parameter set gives precision (the precision of additions,
rotations, multiplications) $\approx 29$ bits and bootstrapping
precision (the precision of the bootstrapping process) $\approx 22$
bits. Overall, 9 multiplicative levels, numbered from 12 to 4, are
available before bootstrapping; we can go down to level 0 if we do not
need to bootstrap (for instance, upon returning). As already
mentioned, bootstrapping restores the full multiplicative budget.

In~Table~\ref{tab:fgb_data}, we give timings (in a single-thread CPU
context) for the main operations at maximum level $\ell = 12$ (recall
that the cost of additions, rotations and multiplications is roughly
linear in $\ell + 1$).

\begin{table}[H]
\caption{\label{tab:fgb_data}Overview of the base FGb parameters in the HEaaN library and times for each HE operation. $\log(QP), N_0, L, \lambda$ denote the bit lengths of the largest modulus,
  the number of slots (recall that the ring degree $N$ is $2N_0$),
  the multiplicative depth,
  and the security parameter. Precision is given in bits. 
  }
\begin{tabular}{|c|c|c|c|c|c|c|}
\hline
{$\log(QP)$} & $N$ & $N_0$ & $L$ & $h$ & $\lambda$ & $\preci\,(\mult ~/ \BTS)$\\
\hline
{1555} & $2^{16}$ & $2^{15}$ & 9 & 192 & 128 & -29 ~/ -22 \\
\hline \hline
\end{tabular}
\begin{tabular}{|c|c|c|c|c|c|}
\hline
\multicolumn{6}{|c|}{Time of each homomorphic operation}\\\hline
  $\enc$& $\dec$ & $\add$ & $\rot$ & $\BTS$ & $\mult~(\textrm{pt}~/ \ct)$ \\\hline
  64ms & 14ms & 2.1ms & 130ms & 17s & 2.6ms ~/ 210ms \\\hline
\end{tabular}
\label{tab:base}
\end{table}

\subsubsection{Designing auxiliary functions}\label{ssec:implem-aux-functions}
Several options are possible for implementing the function $x\mapsto x^{-1/2}$, combining minimax approximants and Newton's method; the former choice minimizes the level usage, while the latter reduces the number of multiplications. We have chosen to use minimax approximants, computed using the Sollya~\cite{Sollya} tool. In the case of the inverse square root, we use the weighted version of the Remez algorithm, in order to obtain a weighted approximant, namely, a polynomial $P$ which minimizes $\|P(x)\sqrt{x} - 1\|_\infty$ rather than $\|P(x) - 1/\sqrt{x}\|_\infty$.
For the exponential function, we have also used minimax approximations. We have chosen degrees of the form $2^t - 1$, as it maximizes accuracy for a fixed-level budget of $t$.

The implementation of approximate renormalization (see~Remark~\ref{rem:approx_renorm_practice}) requires some care. First, it should be pointed out that the input of each of the renormalization steps depends on $n$ (but an implementation valid in dimension $n$ will also be valid in smaller dimensions). Using approximate normalization steps means that a given step returns a value $\lambda$ such that $|\lambda^2 \sum_{i=1}^n y_i^{(j)^2} - 1| \le \alpha$. As a consequence, we have $\sum_{i=1}^n y_i^{(j+1)} \in [1-\alpha, 1+\alpha]$ and, using Lemma.~\ref{le:cs}, the input of the next renormalization step is guaranteed to be in  $[(1-\alpha)^2/n, (1+\alpha)^2]$. Our design thus uses, for Steps $2$ to $(k-1)$, a polynomial such that $|x P(x)^2 - 1|\le \alpha$ over the interval $[(1-\alpha)^2/n, (1+\alpha)^2]$. 

\subsubsection{Bootstrapping strategy in the many ciphertexts case}
As discussed in Section~\ref{sec:parallel}, an efficient implementation of the many ciphertexts case must avoid bootstrapping in the main thread, since this would mean bootstrapping a large number of ciphertexts. In comparison, the bootstrappings of the auxiliary thread are much less expensive, as they can be amortized across the various plaintexts. 

In order to achieve our goal of avoiding bootstrappings in the main thread, a careful schedule of the bootstrapping of the auxiliary thread is required, as the output of the auxiliary thread is injected back into the main thread. 

Our choice is to minimize the level consumption of the main thread:
\begin{itemize}
\item if $\ctO$ has sufficiently high multiplicative budget at Step~6 of~Algorithm~\ref{alg:AuxThread}, so as to allow the full evaluation of~Algorithm~\ref{alg:AuxThread} without bootstrapping, we bootstrap the result of~Algorithm~\ref{alg:AuxThread}; 
\item otherwise, we bootstrap $\ctO$ before or after Step 6 of Algorithm~\ref{alg:AuxThread}. If the level of $\ctO$ at the end of Algorithm~\ref{alg:AuxThread} is lower than the level of the main thread, we bootstrap $\ctO$ again. 
\end{itemize}

Note that the main thread must have one multiplicative level available when entering the auxiliary thread, in order to execute Step~1 of Algorithm~\ref{alg:AuxThread}. 

\subsubsection{Implementation of Algorithm~\ref{alg:Softmax1B}}
The purpose of this algorithm is to limit the number of bootstrappings
of the main thread. In particular, with our parameters, it is possible
to execute~Algorithm~\ref{alg:Softmax1B} on the range $[-128, 0]$
without bootstrapping the main thread. We consider
Algorithm~\ref{alg:Softmax1B} on this range only.

\subsection{Experiments}
The goal of our experiments is both to give timings for the interested
reader, but also to demonstrate that our asymptotic results are
confirmed in practice. We also showcase the performance of
Algorithm~\ref{alg:Softmax1B}.

It should be noted that our experimental setting, in particular the
input range, is derived from our tests on Meta's LLaMa (7B version)
LLM~\cite{LLAMA}. In this setting, the dimension of the Softmax calls
correspond to the number of tokens (the sum of the input  and output size).

\subsubsection{Experimental setting}
All our CPU experiments have been run on an Intel Xeon Silver 4114 CPU at
2.2GHz with 260GB of RAM running Linux. All our timings are given for
a single-thread execution. We shortly discuss GPU timings at the very
end of this section. 

\subsubsection{Input distribution}
We chose to experiment on random inputs. Regarding the distribution of
our inputs, we have chosen a normal distribution centered at $-M/2$
with standard deviation $M/6$ (tailcut so that all elements belong to
$[-M, 0]$), but found that the uniform distribution on $[-M, 0]$ gives
very similar results. Our experiments on LLaMa suggest that this
normal distribution is a suitable model for practical inputs to
Softmax in this LLM, with $M = 128$. Still, we also experimented using
the same distribution for $M = 256$ in order to demonstrate the
scalability properties of our method.

\subsubsection{Error measurements}
\def\errabs{{\mathrm{err}_{\mathrm{abs}}}}
\def\errrel{{\mathrm{err}_{\mathrm{rel}}}}
\def\Algorithm1{{\mathrm{Algorithm1}}}
As our computations use a fixed-point system, we have mostly measured the absolute error, defined by 
\[
\errabs = \| \SM({\mathbf x}) - \Algorithm1({\mathbf x}) \|_\infty,
\]
where for a vector $\mathbf{x}$, the notation $\|\mathbf{x}\|_\infty$ stands for the largest coordinate of $\mathbf{x}$ in absolute value. 
The associated absolute precision is $-\log \errabs$. This choice seemed more adapted to us due to the fact that our underlying arithmetic is fixed-point. 

The alternative would be to measure the {\em relative precision}, which we define by 
\[ \errrel = \frac{\|\SM({\mathbf x}) - \Algorithm1({\mathbf x})\|_\infty}{\|\SM({\mathbf x})\|_\infty}.
\]

We point that $\|\SM({\mathbf x})\|_\infty \ge 1/n$, which shows that the precision difference between the two measures is at most $\log n$ bits. This difference will be large when all values are very close to one another, and small when the maximal value is significantly larger -- the second case being more typical than the first one in applications. 

In our experiments, we measure the worst-case accuracy but also the average accuracy and standard deviation. The results are obtained via experiments on 5000 Softmax in each dimension. The average accuracy in bits is defined as the logarithm of the average absolute error, and similarly for standard deviation. 

\subsubsection{The one-ciphertext case}
We give in Table~\ref{tab:alg1} the accuracy and time (in single-thread CPU) of Algorithm~\ref{alg:Softmax} for computing Softmax in dimension $n$, in the range $[-256, 0]$. 

\begin{table}[H]\caption{\label{tab:alg1}Experimental results for Algorithm \ref{alg:Softmax} for various values of $n$. All inputs lie within the interval $[-256,0]$.
}
\begin{tabular}{|c|c|c|c|c|c|}
\hline
$M$                    & $n$    & time (s) & \multicolumn{3}{|c|}{precision (bits)} \\
& & & worst  case & std. dev. & average\\ 
\hline
\multirow{4}{*}{256} & 128  & 124      & -15.5 &  2.0 &-21.0            \\ \cline{2-6} 
                      & 256  & 142      & -16.6 & 1.5 & -20.5            \\ \cline{2-6} 
                     & 512  & 145     & -15.7 & 1.0 & -19.4            \\ \cline{2-6} 
                     & 1024 & 182      & -15.3 & 0.9 & -18.8          \\ \hline
\end{tabular}
\end{table}

We have also tested our algorithm for $n = 32768$ on one input as a
stress test. Another motivation is that it is the dimension of the
final Softmax call, for the token generation stage of LLaMa.  In that
case, the computing time was 254s, in line with the linear growth in
$\log n$ (we include it on the graph of Figure~\ref{fig:pictime}), and
the absolute precision was $\approx -12.8$ bits in the worst case
(-17.2 on average); however, as this absolute precision might be
meaningless in that case we measured the relative precision and found
it to be $\approx -6.5$ bits. It is not surprising that the accuracy
degrades somehow. Still, a relative precision of $6.5$ bits remain
significant enough for many applications, even though more experiments
would be needed in order to validate this accuracy.

\subsubsection{Accuracy}
Table~\ref{tab:alg1} demonstrates a very good numerical
behaviour of our algorithm. Our precision is limited by the
bootstrapping precision (22 bits); the overall loss of precision is
thus at most of the order of $6$ bits in the worst case, which is
significantly better than even our heuristic estimate of $k + 1.5\log n$ bits. 

Further, part of this precision loss occurs by construction. Indeed,
our design choices for auxiliary functions reflect our target
precision of 16 bits. In particular, the choice of a 19.5 bit precise
approximation for the final square root (see
Appendix~\ref{app:aux_functions}) limits, after squaring, the final
precision to around 18.5 bits in the worst case. 

We have also experimented on~Algorithm~\ref{alg:Softmax1B} and found
it to have similar accuracy compared to Algorithm~\ref{alg:Softmax}.

\subsubsection{Efficiency}
We demonstrate the linear dependency in $\log n$ of the one ciphertext
Algorithm~\ref{alg:Softmax} in the one ciphertext setting; the timings
provided in Table~\ref{tab:alg1}, illustrated in
Figure~\ref{fig:Meta-BTS}, support this result. Timings for $n=512$
are better than expected due to the fact that we are able to use the
same degree as for $n = 256$ for the polynomial approximation of the
auxiliary thread.

\begin{figure}[h]
  \caption{Time as a function of $\log n$.\label{fig:pictime}}
  \Description{This image plots the cost in seconds of our implementation as a function of $\log n$. It illustrates the good fit of the experimental results with theory, which shows a linear cost in $\log n$.}
    \centering
\begin{tikzpicture}
  \begin{axis}[
      width=.42\textwidth, height=.25\textheight, 
    xlabel={$\log$ of input dimension $n$},
    ylabel={time\,(s)},
    xmin=6, xmax=16,
    ymin=100, ymax=260,
    legend pos=north west,
]
\addplot[
    color=black,
    mark=square,
    ]
    coordinates {(7,124)(8,142)(9,145)(10,182)(15,254)};
    \addlegendentry{Algorithm \ref{alg:Softmax}}
\end{axis}
\end{tikzpicture}
\end{figure}

\subsubsection{The many-ciphertext case}\label{sec:5_many_ctxt}
In this paragraph, we provide results for the many-ciphertext version
of Algorithms~\ref{alg:Softmax} and~\ref{alg:Softmax1B} and HETAL.

In order to provide a fair comparison with HETAL in that setting, we
have reimplemented the latter using the ideas described in
Section~\ref{sec:parallel} to tailor it to the many-ciphertexts
case. We provide a comparison based on this implementation.

We have experimented with 1 to 64 simultaneous ciphertexts for Softmax
in dimension 256 over [-128, 0]; as each ciphertext packs~128 such
Softmax, we shall thus compute up to 8192 Softmax at a time.  We point
out that these parameter values are typical of current mainstream
Large Language Models.

Table~\ref{tab:SMmany} gives the timings for the three methods. 

\begin{table}[H]
\caption{Comparison between Algorithms~\ref{alg:Softmax},~\ref{alg:Softmax1B} and HETAL~\cite{HETAL} -- times in s., single-thread CPU.\label{tab:SMmany}}
\begin{tabular}{|c|c|c|c|c|c|c|c|}
\hline
num\_ctxt                       & 1     & 2     & 4     & 8 & 16    & 32    & 64    \\ \hline
Alg.~\ref{alg:Softmax}    & \textbf{130}   & \textbf{142}  & \textbf{169}   & 228   & 357   & 617 & 1080  \\ \hline
Alg.~\ref{alg:Softmax1B}  & 160   & 168   & 179   & \textbf{193}   & \textbf{229}   & \textbf{286} & \textbf{414}  \\ \hline
HETAL~\cite{HETAL} & 388   & 395   & 455   & 605   & 932   & 1670 & 3180  \\ \hline
Best speedup & 3.0 & 2.8 & 2.7 & 3.1 & 4.7 & 5.8 & 7.7 \\ \hline
\end{tabular}
\end{table}

For 64 ciphertexts in dimension 128, we obtain an amortized cost per
Softmax of 0.05s for Algorithm~\ref{alg:Softmax1B}, 0.13s for
Algorithm~\ref{alg:Softmax}, and of~0.39s for HETAL's approach.

Finally, we point that for Algorithm~\ref{alg:Softmax}, the cost of
the bootstrapings is $\approx 80\%$ of the total cost; for
Algorithm~\ref{alg:Softmax1B}, it decreases from $\approx 70\%$ (for 1
ciphertext) to $27\%$ (for 64 ciphertexts). This clearly points to the
main difference between the two approaches and illustrates that the
key of efficiency lies indeed in the idea of avoiding bootstrappings in 
the main thread.

\subsubsection{GPU implementation}
Finally, we have implemented and run a GPU implementation of our
results on an NVIDIA RTX-6000 GPU. We have, in that case, simply
measured the total cost of all Softmax calls in a full run of LLaMa
(7B version) with 128 tokens; in practice, this means 32 times (one
per layer) 32 calls to 128 Softmax of dimension 128 (i.e., we are in
the many-ciphertext variant with 16 ciphertexts), plus a final call to
one Softmax of dimension 32768. The total duration of all these
computations was less than 90 seconds.  Table~\ref{tab:SMmany} shows
that the single-thread CPU variant would require around 7600 seconds
for the same computation; the GPU architecture thus brings a speedup
of a factor $\approx 90$.

\section{Conclusion}
We have presented a new HE-compatible Softmax algorithm based on a
normalize-and-square strategy. Our experiments demonstrate
that it allows real-world computations with GPU timings
that come close to practical applicability. Overall, this work will
help develop applications of HE to privacy-preserving machine learning. 

\section{Acknowledgements}
We wish to thank Jung Hee Cheon for suggesting this research question to
us, and Taekyung Kim and Ernest Ryu for helpful discussions.

\bibliographystyle{ACM-Reference-Format}
\bibliography{bib}

\appendix
\section{Non-linear function evaluation}\label{app:aux_functions}
We recall that we have chosen to focus on mathematical error due to
approximation and to neglect the error coming from polynomial
evaluation.

\subsection{Computing exponentials}
\begin{lemma}\label{le:exp}
When $A, p \rightarrow \infty$, there exists a polynomial $P$ such that
$\max_{x\in [-A, 0]} |P(x)-\exp(x)|\le 2^{-p}$ and $P(x)$ can be evaluated
using  
    \[ 
    \max (\log p, \log A) + O_{p,A}(1)
    \] levels and $O(\max( \sqrt{p}, \sqrt{A})$ ciphertext-ciphertext
    multiplications.
\end{lemma}
\begin{proof}
We use the degree $d$ Taylor expansion at 0, whose error (as an
alternating series) is $\le A^{d+1}/(d+1)! = 2^{-d\ln (d/(eA)) +
  o(d)}$. For $d = \max(p, e^2A)$, this is asymptotically less than
$2^{-p}$.

It then suffices to evaluate a degree $d$ polynomial, which can be performed
using 
$\log d + O(1)$ levels and $O(\sqrt{d})$ ciphertext-ciphertext multiplications.
\end{proof}

\subsection{Computing inverse square roots}
\begin{lemma}\label{le:invsqrt}
For $A, p\rightarrow \infty$ and any fixed constant $\alpha$, there
exists a polynomial $Q$ such that $|Q(x) - 1/\sqrt{x}| \le 2^{-p}$ for
$x\in [1/(\alpha A), \alpha]$ using $(\log A) (1 + o(1)) + O(\log p)$
levels and $O(\sqrt{A\log A} + \log p)$ ciphertext-ciphertext
multiplication.
\end{lemma}
\begin{proof}
    We consider the case of the interval $[1/A, 1]$, the general case being
    similar. 
    
We start by computing a truncated Taylor series of $1/\sqrt{x}$ at
$1$. As the coefficients of this series are $\le 1$ in absolute value,
the error of truncating the series at degree $k$ over $[1/A,1]$ is
$O(A(1-1/A)^k)$. For $k = O(A\log A)$ this can be made smaller than
any constant; e.g., we compute $y_0$ such that $|y_0 \sqrt{x} - 1| <
1/2$.

We then use Newton's inverse square-root iteration under the form $y_{n} = y_{n-1}n(3 - x \, y_{n-1}^2)/2$. This iteration can be evaluated using 2 levels at each step (once $x/2$ is computed) and $4$ multiplcations using the following algorithm:
\begin{algorithm}
$z_1\gets x/2 \cdot y$ \Comment*[r]{$l(z_1)= \max(l(x/2), l(y))+1$ }
$z_2 \gets y \cdot y$ \Comment*[r]{$l(z_2)= l(y)+1$ }
$z_3 \gets 3/2 \cdot y$ \Comment*[r]{$l(z_3)= l(y)+1$} 
\Return($z_3 - z_1 z_2$) \Comment*[r]{level $\max(l(x/2), l(y))+2$}
\end{algorithm}

We have
\[
|y_{n}\sqrt{x} - 1| \le |y_{n-1}\sqrt{x} - 1|^2 |y_{n-1}\sqrt{x}+2|/2.
\]

One can check by induction that $y_{n}\sqrt{x} \le 3/2$, which implies that 
\[
|y_{n}\sqrt{x} - 1| \le \frac{7}{4} |y_{n-1}\sqrt{x} - 1|^2 \le \frac{4}{7} \left(\frac{7}{4} |y_0\sqrt{x}-1|\right)^{2^n}
\le \frac{4}{7} \left(7/8)\right)^{2^n}.
\]

Overall, we use a degree $O(A\log A)$ polynomial approximation for
computing $y_0$, so $\log A (1+o(1))$ levels and $O(\sqrt{A\log A})$
ciphertext-ciphertext multiplications. We then use $O(\log p)$ levels and
$O(\log p)$ ciphertext-ciphertext multiplications to then compute
$1/\sqrt{x}$ to the precision $2^{-p}$. The total number of levels is
thus $(\log A)(1 + o(1)) + O(\log p)$, as claimed, and the total
number of ciphertext-ciphertext multiplications $O(\sqrt{A\log A} + \log p)$.
\end{proof}

The proof given there should not be used as an efficient algorithm. In practice, one should replace the Taylor expansion of degree $O(A \log A)$ by a minimax approximation of degree $\approx \sqrt{A} / 10$ which suffices to bring the Newton iteration within its
zone of quadratic convergence. This implies that in practice, the depth of the computation
is expected to be $\frac{\log A}{2} \cdot (1 + o(1)) + O(\log p)$, and the number of ciphertext-ciphertext multiplications is expected to be $O(A^{1/4})$. 

\section{Proof of Theorem~\ref{th:prec_softmax}}\label{se:pfthm}
We recall the statement of the theorem. 
\begin{theorem}\label{thm:appdx}
Let $M$ be a fixed positive real number. 
We define $k =  \lceil \log M - \log \ln n\rceil$ and
$\varepsilon = 2^{-p}$, and assume that
\[ \max(n^2, 2.9 (2.08 \sqrt{n})^k) \varepsilon \le 1/1000.\] 
Algorithm~\ref{alg:Softmax}, on input $\mathbf{x}\in[-M, 0]^n$, returns a vector $y\in [o(1),1+o(1)]^n$ such that
  \[
     \max_{1\le i\le n} |y_i - \SM(x)_i| \le 2.9 \cdot \left((2.08 \sqrt{n})^k (n+1) + 15.5n^2\right) \varepsilon.
  \]
\end{theorem}

Throughout the analysis, we shall keep the notations of
Algorithm~\ref{alg:Softmax}. When a mathematical quantity bears a
tilde symbol, it stands for the numerical approximation computed in
the execution of the algorithm; when it does not, it is the exact
mathematical value. For example, $\widetilde{\lambda_j}\cdot
\widetilde{z_i}$ differs from $\widetilde{\lambda_j\cdot z_i}$ as an
error occurs in performing the multiplication.

Before proving the theorem, we prove a lemma which gives estimates on
the main quantities used throughout the algorithm; we prove that
$\widetilde{\lambda_j}$ is $O(\sqrt{n})$ except maybe at Step 1, and
that after the first step, the vector $(\widetilde{y_i^{(j)}})_{1\le i\le n}$ is
always approximately normalized (i.e., has a sum close to 1).

As this lemma will be used inductively, we refrain from using $o$ and
$O$ notations, which might imply a hidden exponential growth, and
explicit all error terms. Apart from the term $c c'\sqrt{n}$ which
plays a role in the final result, we overall chose to favour
simplicity over optimality in our estimates.

\begin{lemma}\label{lem:lamy}
Under the assumptions of Theorem~\ref{thm:appdx}, for all $j$, if $c n \ge \sum_{i=1}^n \widetilde{y_i^{(j-1)}} \ge c^{-1}$ with $1\le c \le 1/(n\varepsilon)$, we have, for all $j$, the following inequalities:
\begin{align*}
&\widetilde{\lambda_j} \le c c' \sqrt{n} +\varepsilon,\\
&\left|\sum_{i=1}^{n} \widetilde{y_i^{(j)}} - 1\right| \le 15 c^2c'^2 n^2 \varepsilon, \\
& \widetilde{z_i} \le 1 + 16 c^2c'^{2}n^2\varepsilon
\ \ \ \ \textrm{for all $i$},
\end{align*}
for any $c' \ge (1 - c^2n^2\varepsilon)^{-1/2}$. 
\end{lemma}
\begin{proof}
We have $|\widetilde{y_i^{(j-1)^2}} - \widetilde{y_i^{(j-1)}}^2| \le \varepsilon$, so that
  \begin{equation}\label{eq:trivial}
  \left |\sum_{i=1}^n \widetilde{y_i^{(j-1)^2}} - \sum_{i=1}^n \widetilde{y_i^{(j-1)}}^2\right| \le n\varepsilon.\end{equation}

  Our assumption on the inverse square root function implies that
  \begin{equation}\label{eq:approxsqrt}
  \left|\widetilde{\lambda_j} - \left(\sum_{i=1}^n \widetilde{y_i^{(j-1)^2}}\right)^{-1/2}\right| \le \varepsilon,
  \end{equation}

  From $cn \ge \sum_{i=1}^n \widetilde{y_i^{(j-1)}} \ge c^{-1}$ we deduce
  thanks to the Cauchy-Schwarz inequality that
  $c^2n^2 \ge \sum_{i=1}^n \widetilde{y_i^{(j-1)}}^2 \ge c^{-2}/n$. In
  particular, we have $\sum_{i=1}^n \widetilde{{y_i^{(j-1)}}^2} \ge c^{-2}/n -
  n\varepsilon$, so that
\begin{equation}\label{eq:down}
  \left(\sum_{i=1}^n \widetilde{{y_i^{(j-1)}}^2}\right)^{-1/2}
  \le c\sqrt{n} (1 - c^2 n^2\varepsilon)^{-1/2}
  \le cc'\sqrt{n}. 
  \end{equation}
  
We also have
\begin{equation}\label{eq:up}
  \left(\sum_{i=1}^n \widetilde{{y_i^{(j-1)}}^2}\right)^{1/2}
  \le (c^2n^2 + n\varepsilon)^{1/2} \le 2cn. 
  \end{equation}

These yield bounds on $\widetilde{\lambda_j}$ and
$\widetilde{\lambda_j} \left(\sum_{i=1}^n \widetilde{y_i^{(j-1)^2}}\right)^{1/2} - 1$ as follows. First, using (\ref{eq:approxsqrt}) and (\ref{eq:down}) we obtain 
\[
  \widetilde{\lambda_j} \le cc' \sqrt{n} + \varepsilon,
\]
  as claimed. In the end of the proof, we shall mostly use the following weaker bound:
\begin{equation}\label{eq:lambdaup}
\widetilde{\lambda_j} \le 2cc'\sqrt{n}.
\end{equation}

Then, from (\ref{eq:approxsqrt}) and (\ref{eq:up}), we deduce 
\begin{equation}
\label{eq:normlambda}
\left|\widetilde{\lambda_j} \left(\sum_{i=1}^n \widetilde{y_i^{(j-1)^2}}\right)^{1/2} - 1\right| \le \varepsilon \left(\sum_{i=1}^n \widetilde{y_i^{(j-1)^2}}\right)^{1/2} \le 2cn\varepsilon, 
\end{equation}
from which it follows that 
  \begin{align} \label{eq:normsquared}
  \left|\widetilde{\lambda_j}^2 \left(\sum_{i=1}^n \widetilde{y_i^{(j-1)^2}}\right) - 1\right| & \le 2cn \varepsilon 
  \left( 1 + \widetilde{\lambda_j} \left(\sum_{i=1}^n \widetilde{y_i^{(j-1)^2}}\right)^{1/2}\right), \nonumber \\
  & \le 2cn\varepsilon (2 + 2 c n \varepsilon), \nonumber\\
  & \le 8cn\varepsilon. 
  \end{align}

Similarly, from 
  $
  \left|
  \widetilde{\lambda_j}\cdot \widetilde{y_i^{(j-1)}} - \widetilde{z_i}
  \right| \le \varepsilon,$
we deduce
  \begin{align*}
  \left|
  \widetilde{\lambda_j}^2\cdot \widetilde{y_i^{(j-1)}}^2 - \widetilde{y_i^{(j)}}
  \right| & \le
  \left|
\left(
\widetilde{\lambda_j}\cdot \widetilde{y_i^{(j-1)}} -
\widetilde{z_i} \right)
\left(
  \widetilde{\lambda_j}\cdot \widetilde{y_i^{(j-1)}} + \widetilde{z_i}
    \right)\right| \\
    & \hspace*{2cm} + \left| \widetilde{y_i^{(j)}} - \widetilde{z_i}^2\right|,  \\  
 & \le 
\varepsilon (2\widetilde{\lambda_j} \widetilde{y_i^{(j-1)}} + \varepsilon) + \varepsilon \\ 
& \le 2 \varepsilon\widetilde{\lambda_j} \widetilde{y_i^{(j-1)}} + 2\varepsilon. 
\end{align*}
  
  Summing, we obtain
\[
  \left|
  \sum_{i=1}^n \widetilde{\lambda_j}^2\cdot \widetilde{y_i^{(j-1)}}^2 - \sum_{i=1}^n \widetilde{y_i^{(j)}}
  \right| \le
  2\varepsilon \widetilde{\lambda_j} \sum_{i=1}^n \widetilde{y_i^{(j-1)}} + 2n\varepsilon. 
  \]

From (\ref{eq:lambdaup}) and the assumptions of the lemma, we get
\[\widetilde{\lambda_j} \sum_{i=1}^n \widetilde{y_i^{(j-1)}} \le
2 c^2 c' n^{3/2},\]
from which we deduce
\begin{align}\label{eq:stepntonp1}
  \left|
  \sum_{i=1}^n \widetilde{\lambda_j}^2\cdot \widetilde{y_i^{(j-1)}}^2 - \sum_{i=1}^n \widetilde{y_i^{(j)}}
  \right| & \le 4c^2c' n^{3/2} \varepsilon + 2n\varepsilon, \nonumber\\
  & \le 6 c^2 c' n^{3/2} \varepsilon. 
  \end{align}
  
We finally obtain, using (\ref{eq:trivial}), (\ref{eq:lambdaup}) and (\ref{eq:normsquared}),  
  \begin{align*} 
  \left|
  \sum_{i=1}^n \widetilde{\lambda_j}^2\cdot \widetilde{y_i^{(j-1)}}^2 - 1
  \right|
  & \le 
  \left| \widetilde{\lambda_j}^2 \sum_{i=1}^n \widetilde{y_i^{(j-1)^2}} - 1\right| + n \widetilde{\lambda_j}^2\varepsilon,\\
  & \le 8cn\varepsilon + c^2{c'}^2n^2 \varepsilon, \\
  & \le 9 c^2{c'}^2 n^2 \varepsilon, 
  \end{align*}
which we combine to (\ref{eq:stepntonp1}) to deduce $|\sum_{i=1}^n \widetilde{y_i^{(j)}} - 1| \le 15 c^2{c'}^2 n^2 \varepsilon$, as claimed. 
As $|\widetilde{y_i^{(j)}} - \widetilde{z_i}^2| \le \varepsilon$ we
also have that for all $i$, $\widetilde{z_i} \le \max(1, \widetilde{z_i}^2)
\le 1 + \varepsilon + 15c^2c'^{2} n^2 \varepsilon \le 1 + 16c^2c'^{2}n^2\varepsilon$, from which the last point follows.
\end{proof}

We now turn to the proof of the theorem.

\begin{proof}
As $k = \lceil \log M - \log \ln n\rceil$, we have $2^k \ge M/\ln n$, so
that $x_i \in [-M, 0]$ implies that $x_i / 2^k \in [-\ln n, 0]$ for all $1\le i\le n$. Hence,
under our assumptions on the exponential function,
$\widetilde{y_i^{(0)}} \in [1/n - \varepsilon, 1 + \varepsilon]$.

We choose $n$ large enough so that 
$n^2\varepsilon < 1/1000$. Then, the assumptions of the lemma are
fulfilled for $j = 1$ as soon as $c \ge (1 - n \varepsilon)^{-1}$; for this
it suffices that $c\ge (1 - n^2\varepsilon)^{-1}$. We choose  $c = 1.0158$,
check that we can choose $c' = 1.0006$, and deduce that
\[
1-15.5 n^2 \varepsilon \le
\sum_{i=1}^{n} \widetilde{y_i^{(1)}} \le
1+15.5 n^2 \varepsilon.
\]

For the following loop iterations, 
one checks that we can still use $c = 1.0158$ and $c' = 1.0006$, which
yield again
\begin{equation} \label{eq:bndsum}
1-15.5 n^2 \varepsilon  \le \sum_{i=1}^{n} \widetilde{y_i^{(2)}} \le
1+15.5 n^2 \varepsilon.
\end{equation}
This shows that the induction carries over and that we can assume that
the conclusions of Lemma~\ref{lem:lamy} apply throughout the algorithm\footnote{Note that by making stronger requirements
on $n^2\varepsilon$, the constants $c$ and $c'$ can be brought arbitrarily
close to 1.}
with $c = 1.0158$ and $c' = 1.0006$
With these choices of parameters,  the 
lemma gives, for all $i, j$,
\begin{align}
\widetilde{\lambda_j} \le 1.0158 \cdot 1.0006 \sqrt{n} + \varepsilon \le 1.018 \sqrt{n},\label{eq:bndlam}\\
\widetilde{z_i} \le 1 + 16.53 n^2\varepsilon \le 1.017. \label{eq:bndzi}
\end{align}

We now estimate the error growth of the normalize-and-square strategy. We prove
by induction that for all $j \le k$, there exists $C_j$ such that 
\[
 |\widetilde{y_i^{(j)}} - C_j \exp(x_i/2^{k-j})| \le A_j, 
\]
where $A_j$ is defined, for $j < k$, by $A_0 = \varepsilon$, $A_{j+1} = 2.08 \sqrt{n} A_j + 3.036 \varepsilon$, so that
\begin{align}
A_j & = (2.08\sqrt{n})^j \varepsilon \left(1 + \frac{3.036}{2.08\sqrt{n} - 1} \right) - \frac{3.036}{2.08\sqrt{n} - 1} \varepsilon, \nonumber\\
& \le 2.9\cdot (2.08\sqrt{n})^j \varepsilon. \label{bnd:aj}
\end{align}

This holds for $j = 0$ with $C_0 = 1$ given
our assumptions on the exponential function.

We now fix $j < k$, and assume the induction hypothesis for $j$.
Then, we have, for $1\le i\le n$,
\begin{align}
| \widetilde{z_i} - \widetilde{\lambda_{j+1}} C_j \exp(x_i/2^{k-j})|
&\le |\widetilde{z_i} - \widetilde{\lambda_{j+1}}\widetilde{y_i^{(j)}}| + \nonumber\\
 & 
|\widetilde{\lambda_{j+1}}\widetilde{y_i^{(j)}}-  \widetilde{\lambda_{j+1}} C_j \exp(x_i/2^{k-j})|, \nonumber\\ 
&\le \varepsilon + \widetilde{\lambda_{j+1}} A_j, \label{eq:bndnorminterm}\\
&\le \varepsilon + 1.018 \sqrt{n} A_j, \label{eq:bndnorm}
\end{align}
thanks to \eqref{eq:bndlam}. 

Further, from \eqref{eq:bndlam} and \eqref{eq:bndzi} we derive 
\[
|2 \widetilde{z_i} + \widetilde{\lambda_{j+1}} A_j + \varepsilon| \le
   2.034 + 1.018 \sqrt{n} A_j + 0.001.
   \]
For $j < k$, Inequality~(\ref{bnd:aj}) implies $1.018\sqrt{n} A_j \le 2.9(2.08\sqrt{n})^k \varepsilon < 1/1000$ thanks to our assumption on $\varepsilon$. Overall, we deduce 
\begin{equation}\label{eq:bndsquare}
|2 \widetilde{z_i} + \widetilde{\lambda_{j+1}} A_j + \varepsilon|
 \le 2.036.
 \end{equation}

We have 
\[
| \widetilde{y_i^{(j+1)}} - \widetilde{z_i}^2 |  =
| \widetilde{z_i^2} - \widetilde{z_i}^2 | 
\le  \varepsilon,
\]
from which we deduce the following chain of inequalities: 
\begin{align}
| \widetilde{y_i^{(j+1)}} - (\widetilde{\lambda_{j+1}} C_j)^2 & \exp(x/2^{k-j-1}) | \le \nonumber \\
& \varepsilon + |\widetilde{z_i}^2 - (\widetilde{\lambda_{j+1}} C_j)^2 \exp(x/2^{k-j-1})|,\nonumber \\
& \hspace*{-.1cm}\le
\varepsilon + |\widetilde{z_i} - (\widetilde{\lambda_{j+1}} C_j) \exp(x/2^{k-j})|\cdot \nonumber \\
& \hspace*{2cm} |\widetilde{z_i} + (\widetilde{\lambda_{j+1}} C_j) \exp(x/2^{k-j})|,  \nonumber \\
& \hspace*{-.15cm}\stackrel{\eqref{eq:bndnorminterm}}{\le} \varepsilon + |\widetilde{z_i} - \widetilde{\lambda_{j+1}} C_j \exp(x_i/2^{k-j})| \cdot \nonumber\\
& \hspace*{2cm}|2 \widetilde{z_i} + \widetilde{\lambda_{j+1}} A_j + \varepsilon|, \nonumber \\
& \hspace*{-0.3cm}\stackrel{\eqref{eq:bndnorm}, \eqref{eq:bndsquare}}{\le} \varepsilon + 2.036 (\varepsilon + 1.018 \sqrt{n} A_j) \nonumber\\
& \le 2.08 \sqrt{n} A_j +  3.036 \varepsilon = A_{j+1},
\label{eq:proof_intermed}
\end{align}
which concludes the induction by taking $C_{j+1} = (\widetilde{\lambda_{j+1}} C_j)^2$.
\smallskip

We now bound the final normalization error. For $j = k$, we find, for $1\le i \le n$, 
\[
|\widetilde{y_i^{(k)}} - C_k \exp(x_i)| \le A_k.
\]
Summing, we have 
\[ 
\left|C_k - \frac{\sum_{\ell=1}^n  \widetilde{y_\ell^{(k)}}}{\sum_{\ell=1}^n  \exp(x_\ell)}\right| \le \frac{nA_k}{\sum_{\ell=1}^n \exp(x_\ell)}, 
\]
so that, for all $1\le i \le n$, \def\bx{\mathbf{x}}
\begin{align*}
\left|\widetilde{y_i^{(k)}} - \sum_{\ell=1}^n  \widetilde{y_\ell^{(k)}} \SM(\bx)_i\right| & \le A_k + nA_k \frac{\exp(x_i)}{\sum_{\ell=1}^n \exp(x_\ell)} \\
& \le (n+1)A_k. 
\end{align*}
Finally, using \eqref{eq:bndsum}, we have 
\[
\max_{1\le i\le n}|\widetilde{y_i}^{(k)} - \SM(\bx)_i| \le (n+1)A_k + 15.5 n^2 \varepsilon,
\]
completing the proof. 
\end{proof}

\begin{remark}
The final estimate is highly pessimistic.
Looking at the proof in a more heuristic way, we can notice that the
error induced by the ``normalize and square'' strategy is expected to
be dominated by: 
\def\by{\mathbf{y}^{(0)}}
\begin{align}
  {\mathcal B}_{k,\by} :=  2^k \prod_{j=1}^k \lambda_j^2 y_i^{(j-1)}, \label{eq:bigprod}
\end{align}
where $i$ is the index for which $x_i$ is maximal (without loss of generality,
we assume $i = 1$ in the sequel). 
By induction, ignoring the various evaluation errors, one checks that, for $j\ge 2$, $1\le i\le n$, the following mathematical identities hold: 
\[
\lambda_{j}^2 = \frac{\left(\sum_{\ell=1}^n {y_\ell^{(0)}}^{2^{j-1}}\right)^2}{\sum_{\ell=1}^n {y_\ell^{(0)}}^{2^j}}, y_i^{(j-1)} = \frac{{y_i^{(0)}}^{2^{j-1}}}{\sum_{\ell=1}^n {y_\ell^{(0)}}^{2^{j-1}}}
\]
and that $\lambda_1^2 y_1^{(0)} = y_{1}^{(0)} / \sum_{\ell=1}^n {y_\ell^{(0)}}^2$, 
so that
\[
{\mathcal B}_{k, \by} = 2^k  
    \frac{y_1^{(0)^{2^k-1}}}{\sum_{\ell=1}^n y_\ell^{(0)^{2^k}}} \le \frac{2^k}{y_1^{(0)}} \le n2^k, 
\]
where we have lower bounded the denominator by $y_1^{(0)^{2^k}}$ and
the last inequality follows from our choice for $k$. 
This heuristic analysis suggests that in the worst case scenario the normalize and square strategy should lose at most of the order of $k+ \log(n)$ bits of precision. 

We must also consider the normalization error, related to the error on
$\sum_{i=1}^n y_i^{(k-1)}$, which on average is expected to be
$O(\sqrt{n})$ times the error on $y_i^{(k-1)}$. Overall, we thus expect a total loss of
precision of the order of $k + 3(\log n)/2$ bits, which is already a
pessimistic estimate compared with our experiments, see~Table~\ref{tab:alg1}. 
\end{remark}

\section{Different normalization strategies}\label{se:diffnorm}
For the sake of completeness, we discuss here different normalization strategies. The renormalize-and-square can easily be substituted by a renormalize-and-$t$-th power strategy. We then use
\[
\lambda_j \leftarrow \left(\sum_{i=1}^n y_j^{(i)^t}\right)^{1/t}
\]
as a renormalization factor. 

In practice, the right choice of $k$ goes from $M/2^k \approx \log n$
to $M/t^k \approx \log n$, so we essentially gain a factor $\log t$ on
the total depth of the main thread. However, the renormalization step
becomes more difficult with larger $t$. After the first step, $t$-th
power renormalization gives $\sum_{i=1}^n y_j^{(i)} = 1$, so by
Hölder's inequality we have $\sum_{i=1}^n y_j^{(i)^t} \in [n^{-(t-1)}, 1]$.

This range becomes larger with $t$ so that for fixed $n$, we expect
that the number of levels is multiplied by $t$; overall, we lose a
total factor $t / \log t$ in the level consumption. This shows that
this strategy may only make sense for a limited range of $n$, where
the asymptotic regime does not reflect the experiments. 

\section{Sketch of an analysis of HETAL}\label{app:hetal}
We first outline an analysis of the level of consumption of HETAL.

The first stage of the HETAL algorithm uses a hierarchical maximum
algorithm; starting with $n$ slots, at the first iteration, it
computes an approximation of the maximum of 2 slots, then reduces to
compute the maximum of $n/2$ slots. As such, 
the level cost amounts to $\log n$ hierarchical steps of
comparison. In order to achieve approximation within a constant of the
final maximum, starting with numbers in $[-M, 0]$, it can be proved
that due to the nice properties of the maximum
function \cite{Cheon0K20} used by HETAL, it suffices that each
pairwise comparison returns a maximum within error $O(1)$ to get an
overall maximum within $O(\log n)$. 
The former is equivalent to evaluating a sign function within error
$O(M^{-1})$ on $[-1, 1]$; applying the heuristic
result~{\cite[Corollary 5]{Cheon0K20}} with ``$\alpha = \log M$'', we see
that each pairwise maximum costs essentially $\log M / \log k$
evaluations of a degree $(2k+1)$ polynomial, resulting to a level cost of
$\log M + O(1)$. The total level cost of the maximum computation is thus
$(\log M) (\log n) (1 + o(1))$. The exponential computation and the final
normalization are negligible compared to estimating the maximum.

In the many-ciphertext contexts, starting with $m$ ciphertexts, we can
group the pairwise maxima in $m/2$ ciphertexts, then $m/4$, etc. The
$j$-th level of the hierarchical maximum tree will thus be executed on
$\max(1, m/2^j)$ ciphertexts, for $j\le \log n$. We overall make a
total $\le \log n + 2m - 1$ calls to the maximum, each for a cost of
$\log M (1 + o(1))$, hence an amortized cost of $2 \log M + o_m(1)$
levels per ciphertext. In that case, the exponential computation is no
longer negligible; it requires using domain extension
polynomials \cite{CheonetAl2022}. For the parameters mentioned
in \cite[Algorithm 1]{HETAL}, we find that the level consumption is
$2 \log (M/R) / \log L$, for $1.5 < L < 1.5\sqrt{3}$, so $\ge
2.09 \log M (1 + o(1))$ for fixed $R$.  Overall, we find a global cost
$\ge 4 \log M (1 + o(1))$ per ciphertext.

\end{document}